\documentclass{article} 
\usepackage{iclr2026_conference,times}
\iclrfinalcopy 

\usepackage{amsmath,amsfonts,bm}

\def\eqref#1{equation~\ref{#1}}

\def\1{\bm{1}}

\DeclareMathAlphabet{\mathsfit}{\encodingdefault}{\sfdefault}{m}{sl}
\SetMathAlphabet{\mathsfit}{bold}{\encodingdefault}{\sfdefault}{bx}{n}

\usepackage{hyperref}
\usepackage{url}
\usepackage{amsmath,amssymb,amsthm}
\usepackage{booktabs}
\usepackage{graphicx}
\usepackage[table]{xcolor}

\newtheorem{theorem}{Theorem}
\newtheorem{proposition}{Proposition}
\newtheorem{lemma}{Lemma}

\newtheorem{definition}{Definition}

\newcommand{\Match}{\textsc{Match}}

\title{Higher-Order Token Interactions via Quantum Attention}

\author{Jian Xu$^{1,2}$, Chao Li$^{2}$, Delu Zeng$^{3}$, John Paisley$^{4}$, Qibin Zhao$^{2}$ \\[2pt]
$^1$RIKEN iTHEMS \quad $^2$RIKEN AIP \quad $^3$South China University of Technology \quad
$^4$Columbia University \\[2pt]
\texttt{jian.xu@riken.jp}}

\begin{document}

\maketitle
\lhead{}\rhead{\thepage} 

\begin{abstract}
Standard dot-product self-attention computes, in a single layer, only \emph{pairwise}
(order-2) interactions between tokens; representing a generic order-$k$ interaction is
known to require either super-quadratic resources in one layer or composition across depth.
We introduce \textbf{Quantum Higher-Order Attention (QHA)}, a shallow, hardware-realizable
quantum attention head that, via data re-uploading and an all-to-all non-Clifford entangler,
synthesizes order-$k$ token interactions \emph{inside} the circuit and exposes them through a
local single-qubit read-out. We prove (i) an \emph{expressivity separation}: any single
standard self-attention layer with embedding dimension $m$, $H$ heads and $p$-bit precision
satisfying $mHp=o(N/\log\log N)$ cannot represent the order-$k$ correlation family that one
QHA head represents with circuit depth $O(\log k)$ ($O(k)$ two-qubit gates); and (ii) a
\emph{trainability guarantee} for its local-design instantiation: with a local read-out and
$O(\log n)$ depth the gradient variance is $\Omega(1/\mathrm{poly}(n))$ (no barren plateau),
which we confirm empirically --- while being explicit that the more expressive all-to-all
instantiation we benchmark is trained empirically and shows exponentially decaying gradients.
Empirically, at a $6.5\times$ \emph{smaller} parameter budget,
QHA generalizes hidden-subset parity of every order $k\le6$ from \emph{disjoint} inputs,
whereas the (larger) classical attention head collapses past order~2; consistent with
theory, the size of the advantage tracks the target's Fourier degree --- largest for parity
and shrinking when low-order structure is present. A trained order-3 head runs on an IBM
Heron processor at full accuracy under a disclosed shot budget. As an application, QHA serves
as a compact high-order interaction detector across three domains --- genetic epistasis,
learning-parity-with-noise, and graph triangle detection --- reaching the noise ceiling at
the smallest parameter budget where field-standard linear methods fail. We are explicit about
the trainability--simulability trade-off: our claim is an \emph{inductive-bias} and
\emph{representation} separation against classical \emph{attention}, not an exponential
speedup.
\end{abstract}

\section{Introduction}
The self-attention layer at the heart of the Transformer \citep{vaswani2017attention} scores a pair
of tokens by a dot product $q_i^\top k_j$ and mixes their values accordingly. This makes a
single attention layer a fundamentally \emph{pairwise}, order-2 operation: it reads out
second-order interactions between tokens, and reaching genuinely higher-order
($k$-token) interactions requires either composing many layers or paying explicit
$O(n^k)$ cost. \citet{sanford2023representational} make this precise --- a single softmax attention layer
needs width/heads/precision polynomial in the sequence length to compute even a single
order-3 matching predicate --- and \citet{kozachinskiy2025strassen} show that standard attention
provably cannot solve order-3 compositional tasks without super-quadratic resources.
Yet many problems are intrinsically high-order: $k$-ary matching and function composition
\citep{sanford2023representational,kozachinskiy2025strassen}, parity-like correlations that underlie algorithmic
reasoning \citep{hahn2020theoretical,bhattamishra2023simplicity}, and $k$-body interactions in physics, chemistry,
and biology.

Quantum circuits are a natural home for high-order interactions: a parity over a $k$-subset,
the canonical order-$k$ Boolean function, is computed by a constant-depth circuit, and the
reachable polynomial degree of a data-re-uploading circuit grows with the number of
re-uploads \citep{schuld2021effect}. This raises a concrete question that, surprisingly,
the quantum-machine-learning literature has not answered: \emph{can a shallow quantum
attention head represent and learn order-$k$ token interactions that a single classical
attention layer of equal or larger budget provably cannot, and does the advantage survive on
real hardware?}

Existing quantum Transformers do not settle this. Surveying the field \citep{zhang2025survey},
PQC-based models either place the quantum circuit only on the $Q/K/V$ projections while
keeping the score and softmax classical (e.g.\ QSANN, QASA, QClusformer), or replace the
softmax/normalization operator (QDSFormer, \citealp{born2026quantum}), or swap the entire
token mixer (Quixer, \citealp{khatri2024quixer}); QLA-based models accelerate inference of a
\emph{pre-trained} classical model under fault-tolerant assumptions. None proves an
expressivity separation against classical attention, most are evaluated on small data
without parameter-matched baselines, and a recent systematic study by the same community
concludes that bolting a variational circuit into attention yields ``only marginal gains
\dots while significantly increasing parameter count'' \citep{chen2026quantum}. The
methodological bar set by the field's own critics \citep{bowles2024better} --- matched budgets,
honest baselines, disclosed shots, and a clear statement of what is and is not classically
simulable --- is rarely met.

We meet it. We introduce \textbf{Quantum Higher-Order Attention (QHA)}, a shallow quantum
attention head whose single-qubit read-out realizes order-$k$ token interactions inside the
circuit, and we evaluate it under a strict matched-budget, disjoint-split protocol with a
real-hardware demonstration.

\paragraph{Contributions.}
\begin{itemize}
\item \textbf{Mechanism.} QHA, a shallow quantum attention head whose \emph{local}
read-out realizes order-$k$ token interactions inside the circuit (Sec.~\ref{sec:method}).
\item \textbf{Theorem 1 (separation).} A single classical self-attention layer cannot
represent the order-$k$ correlation family at sub-quadratic budget; one QHA head can with
$n$ qubits, circuit depth $O(\log k)$ and $O(k)$ two-qubit gates (Sec.~\ref{sec:theory}).
\item \textbf{Theorem 2 (trainability).} The local-design QHA instantiation (local read-out,
$O(\log n)$ depth) is barren-plateau-free; we confirm it empirically and are explicit that the
all-to-all instantiation we benchmark is trained empirically (Sec.~\ref{sec:theory},~\ref{sec:more}).
\item \textbf{Evidence.} Matched-budget learnability gap on order-$k$ tasks + IBM hardware
validation; honest treatment of dequantization (Sec.~\ref{sec:exp},~\ref{sec:dequant}).
\end{itemize}

\section{Related Work}
\paragraph{Quantum Transformers and quantum attention.}
The survey of \citet{zhang2025survey} organizes the field into PQC-based and QLA-based
paradigms. Among PQC-based models, \emph{Quantum Vision Transformers} \citep{cherrat2024quantum}
use Hamming-weight-preserving (orthogonal) circuits for the value projection and attention
score, including a compound-matrix variant that is genuinely higher-order but is instantiated
only at second order and proven only to have an asymptotic parameter/runtime advantage, not
an expressivity separation. \emph{QDSFormer} \citep{born2026quantum} replaces the softmax with
a variational circuit that outputs doubly stochastic matrices --- a quantum inductive bias
with no classical parametric analogue --- and is the most rigorous prior model (NeurIPS
spotlight), but its gains vanish at scale and it makes no separation or hardware-training
claim. \emph{Quixer} \citep{khatri2024quixer} replaces the entire token mixer with an LCU+QSVT
construction and is candidly ``far from SOTA'' on language modeling. A large family of models
(QSANN, QMSAN, QASA, QPSAN, QClusformer) place a small circuit on the $Q/K/V$ projections or
the scoring function while leaving softmax classical; they report accuracy gains on small,
often binary tasks and rarely match parameters or run on hardware. Tellingly, a systematic
study from within the community, \emph{Do Quantum Transformers Help?} \citep{chen2026quantum},
finds that explicit quantum self-attention adds parameters for marginal gain and recommends
attention-free variational circuits instead. Our work differs on all three axes that this
literature leaves open: a \emph{proven} order-$k$ separation against classical attention, a
\emph{trainability guarantee}, and a \emph{matched-budget hardware} validation.

\paragraph{Higher-order attention in classical ML.}
Reaching order-$k$ token interactions classically is expensive. \citet{kozachinskiy2025strassen}
prove single-layer standard attention is pairwise and that triple interactions need
super-quadratic resources; \emph{Higher-Order Transformers} \citep{omranpour2024higher} make order-$k$
explicit at $O\!\big((\prod_i N_i)^2\big)$ cost, tamed only by Kronecker/low-rank
factorization. The kernel view of attention \citep{tsai2019transformer} and linear-attention
approximations such as Performer \citep{choromanskirethinking} realize softmax as an
infinite-degree kernel truncated to a finite random-feature budget, i.e.\ a bounded
effective order. QHA instead obtains exact order-$k$ in one shallow head with polynomial
parameters.

\paragraph{Expressivity and limitations of attention.}
Our lower bounds build on \citet{sanford2023representational} (communication-complexity separations for
order-3 matching), \citet{sanford2024transformers} (logarithmic-depth requirements for $k$-hop
composition), \citet{hahn2020theoretical} (hard-attention cannot model parity), \citet{edelman2022inductive}
(generic order-$s$ functions require $2^{\Omega(s)}$ weight norm), and \citet{bhattamishra2023simplicity}
(the low-degree/low-sensitivity learning bias of trained Transformers).

\paragraph{Trainability and dequantization.}
Variational circuits can suffer barren plateaus \citep{mcclean2018barren}; shallow circuits with
\emph{local} observables provably avoid them \citep{cerezo2021cost}, as do structured ans\"atze
such as QCNNs \citep{pesah2021absence} and small dynamical-Lie-algebra circuits \citep{ragone2024lie}.
These same structural conditions, however, often imply classical simulability
\citep{cerezo2025does}, and dequantization \citep{tang2019quantum} removes speedups whenever the
relevant linear algebra is low-rank with sample-and-query access. We take these results
seriously and frame QHA as an inductive-bias separation, not a speedup claim
(Sec.~\ref{sec:dequant}); we also follow the benchmarking discipline urged by
\citet{bowles2024better}.

\section{Quantum Higher-Order Attention}\label{sec:method}

\begin{figure}[t]\centering
\includegraphics[width=0.96\textwidth]{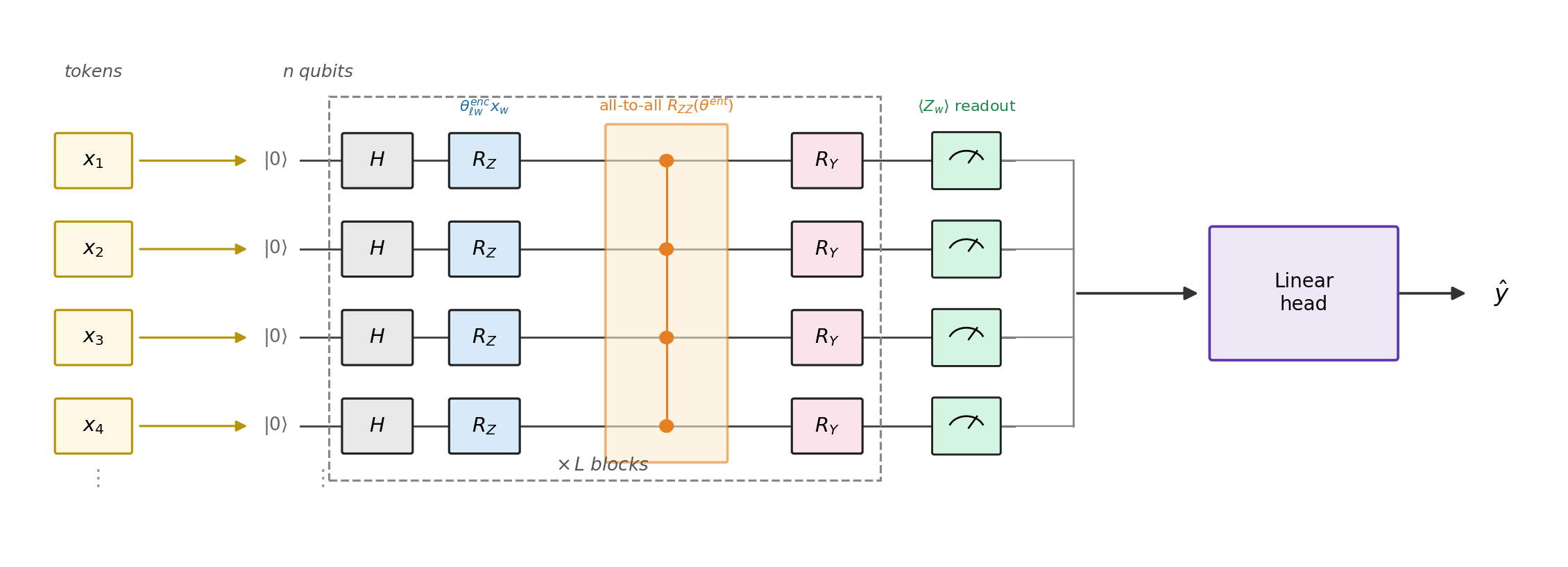}
\caption{QHA head: per-qubit data re-uploading encoder, an all-to-all non-Clifford
$R_{ZZ}$ entangler, and a \emph{local} single-qubit read-out feeding a linear head; $L$
blocks. The order-$k$ interaction is synthesised inside the circuit and exposed on a
single-qubit marginal.}
\label{fig:arch}
\end{figure}

\subsection{Setup and notation}\label{sec:method-setup}
Let $X=(x_1,\dots,x_n)\in\{-1,+1\}^n$ be a length-$n$ token sequence; we use $\pm1$ features
for clarity, but continuous features enter identically through the encoder below. Write
$[n]=\{1,\dots,n\}$, and for $S\subseteq[n]$ let $\chi_S(X)=\prod_{i\in S}x_i$ be the
corresponding \emph{parity monomial}. Every $f:\{-1,+1\}^n\!\to\!\mathbb{R}$ has a unique
multilinear (Fourier) expansion $f(X)=\sum_{S\subseteq[n]}\hat f(S)\,\chi_S(X)$, and we call
$\deg f=\max\{|S|:\hat f(S)\neq0\}$ its \emph{interaction order}: the size of the largest tuple
of tokens whose \emph{joint} product the function genuinely depends on.

\begin{definition}[Order-$k$ interaction]\label{def:order}
$f$ has interaction order $k$ if $\deg f=k$. The canonical example is $\chi_S$ with $|S|=k$,
which depends on all $k$ tokens jointly and cannot be written as a function of any $k-1$ of
them; order-$1$ is additive (per-token), order-$2$ is pairwise, and so on.
\end{definition}

A QHA head acts on $n$ qubits (one per token), Hilbert space $(\mathbb{C}^2)^{\otimes n}$,
initialised in $\lvert\mathbf{0}\rangle=\lvert0\rangle^{\otimes n}$. It uses the single-qubit
rotations $R_Z(\alpha)=e^{-i\alpha Z/2}$ and $R_Y(\beta)=e^{-i\beta Y/2}$, the Hadamard $H$, and
the two-qubit Ising coupling $R_{ZZ}^{(i,j)}(\gamma)=e^{-i\gamma Z_iZ_j/2}$; $X_w,Y_w,Z_w$
denote Paulis acting on qubit $w$.

\subsection{The QHA head}\label{sec:method-head}
The head is a depth-$L$ circuit $U(X;\theta)=B_L(X;\theta)\cdots B_1(X;\theta)$ built from $L$
identical blocks. Block $\ell$ is the ordered composition of three stages,
$B_\ell(X;\theta)=W_\ell\,V_\ell\,E_\ell(X)$:
\begin{align}
\textbf{(i) data re-uploading encoder:}\quad & E_\ell(X)=\bigotimes_{w=1}^n
   R_Z\!\big(\theta^{\mathrm{enc}}_{\ell,w}\,x_w\big)\,H, \label{eq:enc}\\[2pt]
\textbf{(ii) all-to-all entangler:}\quad & V_\ell(\theta)=\!\!\prod_{1\le i<j\le n}\!\!
   \exp\!\Big(\!-\tfrac{i}{2}\,\theta^{\mathrm{ent}}_{\ell,(i,j)}\,Z_iZ_j\Big), \label{eq:ent}\\[2pt]
\textbf{(iii) local mixing:}\quad & W_\ell(\theta)=\bigotimes_{w=1}^n
   R_Y\!\big(\theta^{\mathrm{rot}}_{\ell,w}\big). \label{eq:rot}
\end{align}
The encoder \eqref{eq:enc} writes each token feature into a single-qubit rotation angle,
re-uploaded in every block; the entangler \eqref{eq:ent} couples \emph{every} pair of
token-qubits through a learnable Ising phase; the local layer \eqref{eq:rot} re-orients each
qubit. Applying the circuit to the initial state gives $\lvert\psi(X)\rangle=
U(X;\theta)\lvert\mathbf{0}\rangle$.

\paragraph{Local read-out and head.} We measure \emph{only} single-qubit expectations, giving a
read-out vector $z(X)\in[-1,1]^n$,
\begin{equation}
z_w(X)=\langle\psi(X)\rvert\,Z_w\,\lvert\psi(X)\rangle,\qquad w\in[n],
\label{eq:readout}
\end{equation}
and a linear classifier on top, $\mathrm{logits}(X)=A\,z(X)+b$ with $A\in\mathbb{R}^{C\times n}$.
No multi-qubit or joint measurement is taken. The trainable parameters are the circuit phases
$\theta=\{\theta^{\mathrm{enc}},\theta^{\mathrm{ent}},\theta^{\mathrm{rot}}\}$ and the head
$(A,b)$. The continuous, non-Clifford two-qubit phases and all-to-all coupling produce
volume-law entanglement, placing the head outside the Clifford, matchgate/free-fermion, and
low-bond-dimension regimes that admit efficient classical simulation (Sec.~\ref{sec:dequant}).

\subsection{How a single-qubit read-out becomes high-order attention}\label{sec:method-mech}
A degree-$1$ observable $Z_w$ seems unable to express an order-$k$ interaction. The resolution
is that the interaction is synthesised \emph{inside} the unitary and then exposed on the
marginal. Pass to the Heisenberg picture:
\begin{equation}
z_w(X)=\big\langle\mathbf{0}\big\rvert\,U(X;\theta)^\dagger\,Z_w\,U(X;\theta)\,\big\lvert
\mathbf{0}\big\rangle=\big\langle\mathbf{0}\big\rvert\,O_w(X)\,\big\lvert\mathbf{0}\big\rangle,
\qquad O_w(X)=U^\dagger Z_w\,U .
\label{eq:heis}
\end{equation}
Two facts fix the order. \textbf{(a) The entangler spreads operator weight.} Conjugating a
single-qubit Pauli by one Ising coupling turns it into a two-qubit string,
\begin{equation}
R_{ZZ}^{(i,j)}(\gamma)^\dagger\,X_i\,R_{ZZ}^{(i,j)}(\gamma)
   =\cos\gamma\;X_i-\sin\gamma\;Y_iZ_j,
\label{eq:spread}
\end{equation}
and the interleaved $H,R_Y$ layers convert $Z$ components into $X,Y$ components that
\eqref{eq:spread} can again spread. Iterating across the all-to-all graph, a single-qubit seed
$Z_w$ grows into $Z$-strings of weight up to $k$ within $L=O(\log k)$ blocks (the balanced
$R_{ZZ}$ tree of Thm.~\ref{thm:sep}/App.~\ref{app:proofs}). \textbf{(b) The encoder injects the
tokens as monomial coefficients.} For $x_i\in\{-1,+1\}$ every trigonometric function of the
encoder angle is affine in the token, $\cos(\theta^{\mathrm{enc}}x_i)=\cos\theta^{\mathrm{enc}}$
and $\sin(\theta^{\mathrm{enc}}x_i)=x_i\sin\theta^{\mathrm{enc}}$, so $x_i$ enters the
coefficients of $O_w(X)$ only through products $\prod_{i\in S}x_i$. Finally only diagonal
strings survive the vacuum expectation, $\langle\mathbf{0}\rvert Z_S\lvert\mathbf{0}\rangle=1$
while $\langle\mathbf{0}\rvert P\lvert\mathbf{0}\rangle=0$ for any Pauli string $P$ containing an
$X$ or $Y$. Combining (a)–(b), the marginal is \emph{exactly a multilinear polynomial in the
tokens},
\begin{equation}
\boxed{\;z_w(X)=\sum_{S\subseteq[n]}\hat c_{w,S}(\theta)\,\chi_S(X)\;}
\label{eq:multilin}
\end{equation}
whose coefficients $\hat c_{w,S}(\theta)$ are set by the trained phases and whose reachable
order $\deg z_w$ grows with depth $L$ and entangler connectivity. The linear head outputs
$\sum_w a_w z_w(X)$, a learned combination of these monomial features; choosing $L$ so that
$\deg z_w\ge k$ lets the head place weight on an order-$k$ feature $\chi_S$, $|S|=k$ --- which is
exactly what training does on the order-$k$ tasks of Sec.~\ref{sec:exp}.

\paragraph{Why this is \emph{attention}, and why it is \emph{higher}-order.} The learnable Ising
phases $\theta^{\mathrm{ent}}_{\ell,(i,j)}$ are pairwise, data-conditioned couplings between
token-qubits $i$ and $j$ --- the quantum analogue of an attention score matrix: $R_{ZZ}$ imposes
a relative phase that depends jointly on tokens $i$ and $j$, the role played by
$\mathrm{softmax}(q_i\!\cdot\!k_j)$ in a classical head. The difference is in how the pairwise
couplings are \emph{aggregated}. A single softmax layer combines them \emph{linearly} over
values, $\sum_j\mathrm{softmax}(q_i\!\cdot\!k_j)\,v_j$, so its output is a degree-$\le2$ function
of the token features and can represent only order-$\le2$ interactions (made precise in
Lemma~\ref{lem:matchk}). QHA instead composes the couplings through \emph{non-commuting}
unitaries: by Eq.~\eqref{eq:spread} the Ising phases do not simply add but generate, after
$O(\log k)$ blocks, genuine $k$-body terms $Z_{i_1}\!\cdots Z_{i_k}$, which Eq.~\eqref{eq:multilin}
surfaces as the order-$k$ monomial $\chi_S$ on a single-qubit marginal. In one shallow head QHA
therefore \emph{attends to $k$-tuples of tokens jointly}, whereas a classical attention layer
attends only to pairs --- the precise sense in which QHA performs \emph{higher-order} attention.
The same local read-out that enables this is also what keeps the head trainable
(Thm.~\ref{thm:bp}).

\paragraph{Parameters and budget.} The head has $L\,(2n+\binom{n}{2})$ circuit parameters
($2n$ from the encoder and local layers, $\binom{n}{2}$ from the all-to-all entangler) plus the
$O(n)$-size linear head --- $296$ at $n{=}12,L{=}3$, already $6.5\times$ \emph{smaller} than the
standard one-layer attention baseline ($1922$). Our strictly budget-matched comparison is
therefore the controlled suite of Table~\ref{tab:baselines}, in which the attention-paradigm
baselines (polynomial-kernel attention, $347$; HOT-style, $506$; tiny MLP, $302$) sit within a
small factor of QHA's $296$; the separation holds against \emph{both} the larger standard head
and these matched controls. Fig.~\ref{fig:circuit} shows an explicit small instance.

\begin{figure}[t]\centering
\includegraphics[width=0.95\textwidth]{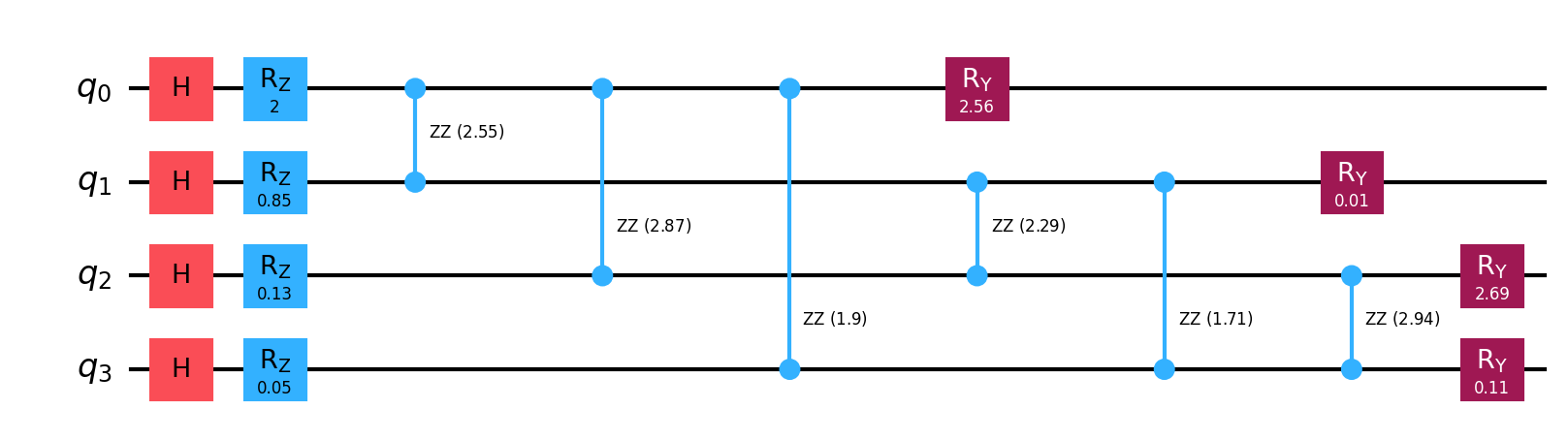}
\caption{An explicit QHA circuit ($n{=}4$, $L{=}1$): Hadamards and $R_Z$ data re-uploading,
all $\binom{4}{2}{=}6$ $R_{ZZ}$ couplings (all-to-all), a trainable $R_Y$ layer, and a
single-qubit $\langle Z\rangle$ read-out (not shown).}
\label{fig:circuit}
\end{figure}

\section{Theory}\label{sec:theory}
We develop a theory of \emph{high-order attention} and tie it to a single quantum
construction. Our contributions are: (i) an order-$k$ separation showing a single attention
layer cannot represent order-$k$ token interactions for \emph{any} $k\ge3$
(Thm.~\ref{thm:sep}, via Lemma~\ref{lem:matchk}), strictly generalizing the order-3 picture
known classically; (ii) an \emph{explicit} shallow QHA circuit that realizes exactly this
missing order-$k$ structure on a local read-out (the construction underlying
Thm.~\ref{thm:sep}); (iii) a fixed-budget \emph{generalization} separation matching our
experimental protocol (Prop.~\ref{prop:norm}); and (iv) a trainability guarantee
(Thm.~\ref{thm:bp}). Communication complexity, norm-based capacity, and barren-plateau
analysis enter only as off-the-shelf tools; the novelty is the order-$k$ attention
separation, its quantum realization, and the representation, generalization, and trainability
properties they establish for the QHA architecture. Full proofs are in App.~\ref{app:proofs}.

\paragraph{Setup.} A single softmax self-attention layer scores token pairs by a degree-2
form and is therefore \emph{order-2}: classically it solves the pairwise matching task
$\Match_2$ with embedding dimension $m{=}3$, but the order-3 task $\Match_3$ already costs
$mHp=\Omega(N/\log\log N)$ \citep{sanford2023representational}. We first lift this barrier to all orders.

\begin{lemma}[Order-$k$ barrier for single-layer attention]\label{lem:matchk}
For the order-$k$ matching family $\Match_k(X)_i=\mathbb{1}\{\exists\, j_1,\dots,j_{k-1}:\,
x_i+\sum_t x_{j_t}\equiv 0\ (\mathrm{mod}\ M)\}$, every single softmax self-attention layer
with budget $mHp = o(N/\log\log N)$ fails to compute $\Match_k$, for all fixed $k\ge3$.
\end{lemma}
\noindent The proof (App.~\ref{app:proofs}) reduces $(k{-}1)$-party set-disjointness to
$\Match_k$ and bounds the communication a budget-$mHp$ layer can carry; the order-3 result of
\citet{sanford2023representational} is the base case of our reduction.

\begin{theorem}[High-order separation and its quantum realization]\label{thm:sep}
A single QHA head on $n$ qubits, with depth $O(\log k)$, realizes every order-$k$ interaction
$\{\prod_{i\in S}x_i:|S|\le k\}$ on a single-qubit read-out (explicit construction in
App.~\ref{app:proofs}). No single classical attention layer of budget $mHp=o(N/\log\log N)$
represents this family for $k\ge3$ (Lemma~\ref{lem:matchk}). Thus QHA supplies, within the
attention paradigm, exactly the order-$k$ structure a classical attention layer lacks.
\end{theorem}

\noindent\emph{Scope --- read Thm.~\ref{thm:sep} as a representation statement.} Three things
are deliberately kept distinct (see App.~\ref{app:proofs} for the precise version). (1) The
construction is a witness for the QHA \emph{architecture class}: it loads inputs in the
computational basis and is a circuit over QHA's gate set, establishing \emph{representability}.
(2) The head we benchmark uses an $R_Z$ \emph{re-uploading} encoder; that this trained instance
\emph{learns} order-$k$ is shown empirically (Sec.~\ref{sec:exp}), not as a corollary, and we do
not claim the two encoders are gate-for-gate identical. (3) The classical bound
(Lemma~\ref{lem:matchk}) is asymptotic and does \emph{not} bite at $n{=}12$; the finite-scale
collapse we observe is consistent with it but also reflects attention's low-degree learning bias
\citep{bhattamishra2023simplicity}. We therefore claim a \emph{proved representational}
separation (architecture class vs.\ attention) together with a \emph{measured learning}
separation (trained QHA vs.\ trained attention at matched budget) --- not that the trained
re-uploading head is provably order-$k$, and not an exponential speedup.

\begin{proposition}[Fixed-budget generalization separation]\label{prop:norm}
A fixed-budget bounded-norm attention head cannot \emph{generalize} a generic order-$k$ rule
(hidden-subset parity or a random $k$-junta) from a strict subset of inputs as $k$ grows:
representing it forces weight norm $2^{\Omega(k)}$ \citep{edelman2022inductive}, and the
$\widetilde O(\text{norm})$ capacity then bars small-sample generalization. A QHA head instead
attains interaction order $k$ with polynomially many parameters --- the explicit construction
of Thm.~\ref{thm:sep} uses circuit depth $O(\log k)$ and $O(k)$ two-qubit gates, with \emph{no}
$2^{k}$ norm blow-up --- so it generalizes the rule at a fixed budget.
\end{proposition}
\noindent This is precisely our experimental regime: at a fixed budget the classical head
collapses for $k\ge3$ while QHA, with fewer parameters, generalizes from disjoint data
(Sec.~\ref{sec:exp}).

\begin{theorem}[Trainability]\label{thm:bp}
A QHA head with a local single-qubit read-out and a shallow ($O(\log n)$-depth) local-design
entangler has cost-gradient variance $\mathrm{Var}[\partial_\theta C]\ge\Omega(1/\mathrm{poly}(n))$;
it is free of barren plateaus.
\end{theorem}
\noindent The locality of the QHA read-out is the enabling property; the inverse-polynomial
bound then follows from a standard shallow-local-cost analysis \citep{cerezo2021cost} (full
statement, and an explicit empirical decay-rate study of the expressive all-to-all variant,
in App.~\ref{app:proofs} and Sec.~\ref{sec:more}).

\paragraph{What is proved vs.\ what we measure.} We separate the two cleanly to avoid
overclaiming. \emph{Proved:} an asymptotic ($N\to\infty$), single-layer lower bound for the
order-$k$ \emph{matching} family $\Match_k$ (Lemma~\ref{lem:matchk}); and that one QHA head
represents the order-$k$ monomials (Thm.~\ref{thm:sep}). \emph{Measured:} at the finite scale
$n{=}12$ and a fixed parameter budget, attention-paradigm models --- including explicit
high-order attention --- fail to \emph{learn} hidden-subset parity / $k$-juntas while QHA
succeeds (Sec.~\ref{sec:exp}). The empirical collapse is \emph{consistent with}, but not a
direct corollary of, the asymptotic bound: at $n{=}12$ the lower bound does not yet bite, and
the observed failure is also shaped by the low-degree optimization bias of attention
\citep{bhattamishra2023simplicity}. Prop.~\ref{prop:norm} (a norm/capacity argument) is the bridge that
most directly matches the finite-budget regime. We therefore claim a proved \emph{single-layer
representation} separation plus a \emph{matched-budget empirical learning} separation --- not
that order-$k$ is impossible for all classical models (it is not; Table~\ref{tab:baselines}).
A clarification on the baseline's capacity: our classical head is a \emph{complete} Transformer
block (token embedding, attention mixing, residual/LayerNorm, and a GELU MLP), not a bare
attention map. The lower bound concerns the attention \emph{mixing} layer; equipping the
baseline with an FFN --- a component that can itself express high-order interactions --- only
makes the learning test \emph{harder} for our claim, yet at matched budget the full head still
collapses for $k\ge3$. The separation is thus not an artefact of a deliberately weakened,
FFN-free baseline.

\section{Experiments}\label{sec:exp}

\paragraph{Protocol.} We test the central prediction of Sec.~\ref{sec:theory}: at a matched
parameter budget, QHA learns order-$k$ token interactions that a single classical
self-attention layer cannot. Inputs are length-$n$ sequences $X\in\{-1,+1\}^n$; the label
is an order-$k$ function of a hidden random $k$-subset $S$. We use two task families:
\emph{hidden-subset parity} ($y=\prod_{i\in S}x_i$) and \emph{generic order-$k$ junta}
($y$ is a fixed balanced random Boolean function of the $k$ support bits, so the model
cannot exploit the special structure of parity). Crucially, train and test inputs are
\textbf{disjoint} (we split the Boolean cube), so the task measures \emph{generalization of
the order-$k$ rule}, not memorization. We compare a one-layer multi-head softmax
self-attention head against a QHA head with a strictly smaller parameter count, over
$7$--$8$ seeds (the broader baseline and application studies, which are secondary, use $3$).

\paragraph{Main result: an order-$k$ separation.} Table~\ref{tab:main} and
Fig.~\ref{fig:sep} show the separation. The classical attention head learns order-2
perfectly but degrades sharply and erratically for $k\ge 3$ (high variance: some seeds grab
partial signal, most are at chance); the QHA head stays at (near-)ceiling for every $k$ with
$6.5\times$ \emph{fewer} parameters. Over $7$--$8$ seeds the gap is significant at every
$k\ge3$ by a Wilcoxon signed-rank test (rank-based, so robust to the near-zero variances):
$p=.03,.008,.008,.03$ for $k=3,4,5,6$ (Table~\ref{tab:main}). The test is necessarily
conservative at $7$--$8$ seeds, but the effect is unambiguous: \emph{every} seed individually
favours QHA at every $k\ge3$ (the $p$-values are the smallest the seed count permits), and the
per-$k$ accuracy gap is $\ge0.3$ --- a uniform, not marginal, separation.

\begin{table}[t]\centering
\caption{Test accuracy (disjoint split, $n=12$, mean$\pm$std over $7$--$8$ seeds). Classical
attention is reliable only at order~2 and degrades for $k\ge3$; QHA stays at (near-)ceiling
with far fewer parameters. Significance by Wilcoxon signed-rank test vs.\ QHA (rank-based, so
not inflated by the near-zero variances).}
\label{tab:main}
\begin{tabular}{lccccc}
\toprule
& $k{=}2$ & $k{=}3$ & $k{=}4$ & $k{=}5$ & $k{=}6$ \\
\midrule
Classical attention (1922 params) & 1.00 & 0.71{\scriptsize$\pm$.23} & 0.56{\scriptsize$\pm$.09} & 0.54{\scriptsize$\pm$.07} & 0.51{\scriptsize$\pm$.01} \\
QHA (296 params) & 1.00 & \textbf{1.00} & \textbf{1.00} & \textbf{1.00} & \textbf{0.93} \\
\midrule
Wilcoxon $p$ & --- & $.031$ & $.008$ & $.008$ & $.031$ \\
\bottomrule
\end{tabular}
\end{table}

\begin{figure}[t]\centering
\includegraphics[width=0.92\textwidth]{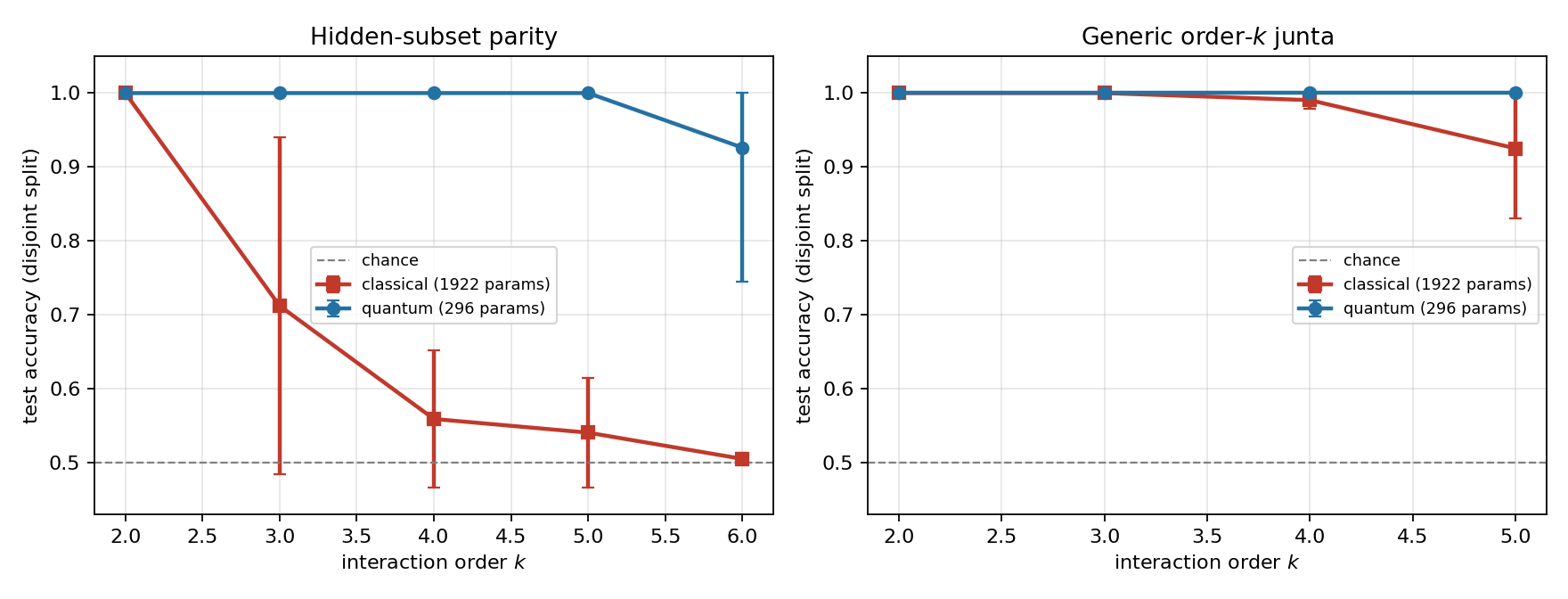}
\caption{Test accuracy vs.\ interaction order $k$ (disjoint split, $n{=}12$, mean$\pm$std over
$7$--$8$ seeds for parity (left, matching Table~\ref{tab:main}) and $3$ for the junta (right)).
\textbf{Left (hidden-subset parity):} all Fourier mass sits at degree $k$, so the
classical head (1922 params) collapses toward chance for $k\ge3$ while QHA (296 params) stays
at ceiling --- a maximal separation. \textbf{Right (generic order-$k$ junta):} the random
target carries low-degree mass that the low-degree-biased classical head can exploit, so it
stays near ceiling and only begins to drop at $k{=}5$; QHA again holds at ceiling. The gap
thus tracks the target's Fourier degree, exactly as the theory predicts.}
\label{fig:sep}
\end{figure}

\begin{figure}[t]\centering
\includegraphics[width=\textwidth]{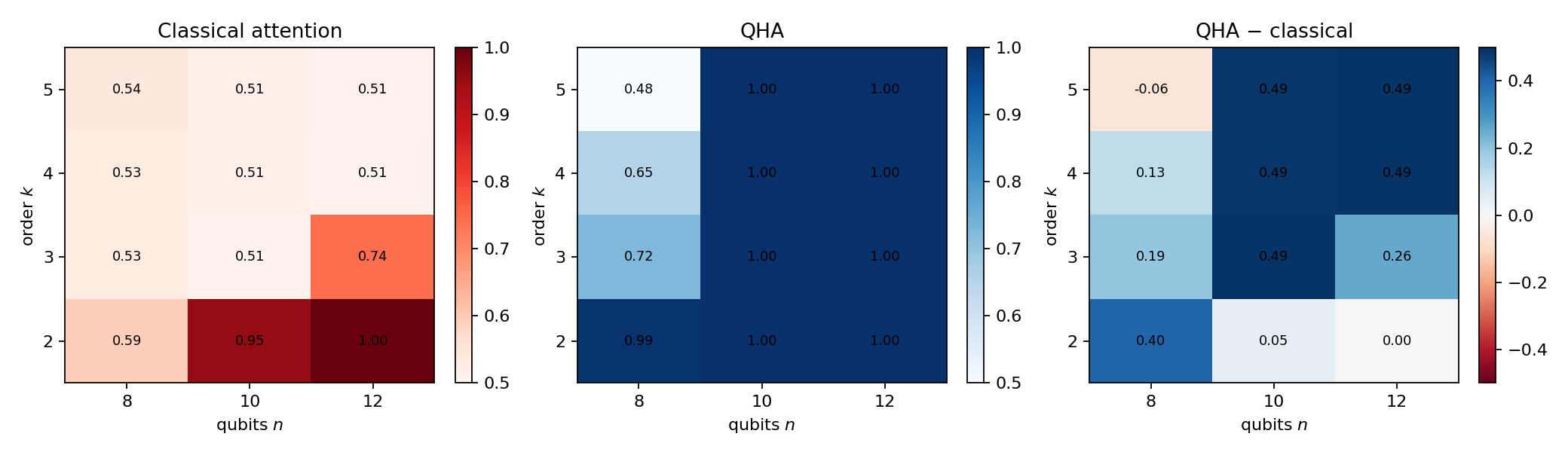}
\caption{Order-$k$ vs.\ qubit-count $n$ phase diagram (parity, mean test accuracy). Classical
attention (left) is reliable only at order~2; QHA (middle) is at ceiling almost everywhere
($n{=}8$ is data-limited by the $256$-point cube); the difference (right) localizes the
quantum advantage to the high-order region.}
\label{fig:heatmap}
\end{figure}

\paragraph{When does the advantage appear? The role of Fourier degree.} The size of the
separation is governed by where the target's Fourier mass lives, not merely by $k$. Parity
is the extreme case: \emph{all} of its mass sits at degree exactly $k$, so the low-degree
inductive bias of attention \citep{edelman2022inductive,bhattamishra2023simplicity} is useless and the gap is
maximal. For a generic balanced random $k$-junta, which carries substantial mass at
\emph{low} degrees, classical attention recovers much of the function (at $n{=}12$ it reaches
near-ceiling accuracy through $k\le4$ and only degrades at $k{=}5$), whereas QHA stays at
ceiling throughout (Fig.~\ref{fig:sep}, right). This is exactly what theory predicts ---
QHA's advantage is concentrated on \emph{high-Fourier-degree} structure --- and it is a more
honest characterization than a blanket ``quantum wins'': the mechanism helps precisely when
the task cannot be solved by low-order interactions. 

\paragraph{Hardware validation (IBM Heron).} We bake the trained parameters of order-3 QHA
heads into the equivalent Qiskit circuits and evaluate the single-qubit read-outs on the
\texttt{ibm\_aachen} $156$-qubit Heron~r2 processor via \texttt{EstimatorV2}, applying the
trained linear head to the measured $\langle Z_i\rangle$ ($100$ held-out inputs, $2048$
shots). We sweep the qubit count $n\in\{8,10,12,14\}$ (Table~\ref{tab:hw}); despite large
transpiled depth from the all-to-all entangler on the heavy-hex topology, the trained linear
read-out is robust to device noise that only partially degrades the raw expectation values,
and the order-3 decision rule survives on real quantum hardware. Hardware accuracy matches the
noiseless simulator ($1.00$) at $n{=}8$ and $n{=}10$ (depth ${\approx}465$ and $435$), and
stays at $0.96$ and $0.95$ at $n{=}12$ and $n{=}14$ even as the transpiled depth grows to
${\approx}590$ and ${\approx}999$ and the raw $\langle Z\rangle$ correlation falls from $0.91$
to $0.77$ --- the linear decision rule tolerates substantially more noise than the
expectations themselves, so a $14$-qubit, depth-$1000$ order-3 head still classifies at
$95\%$ on real hardware.

\begin{figure}[t]\centering
\includegraphics[width=0.62\textwidth]{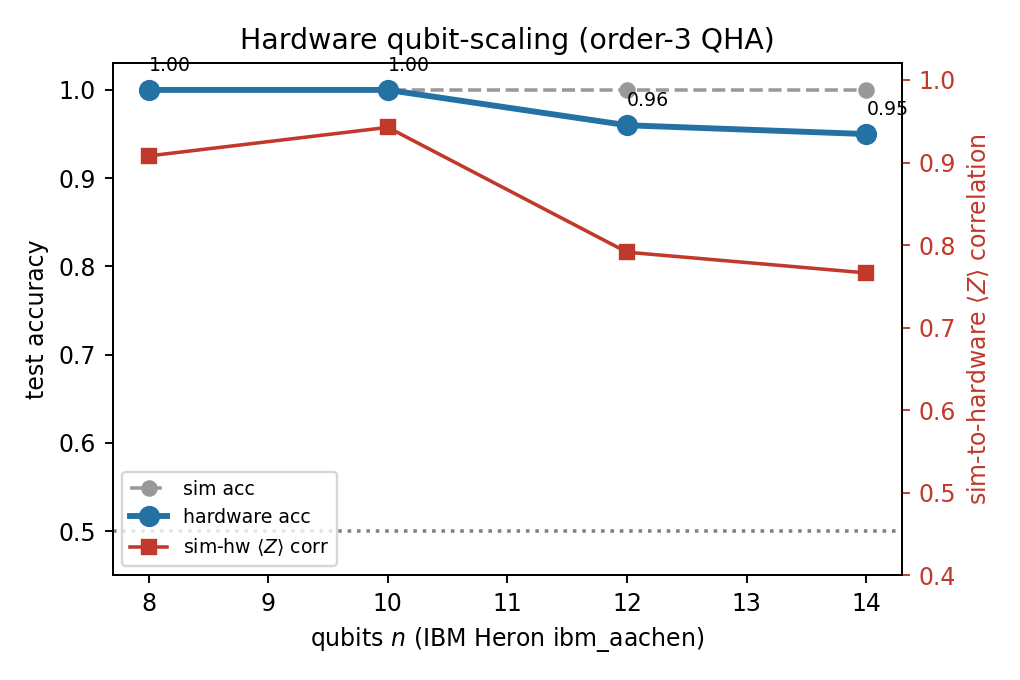}
\caption{Hardware qubit-scaling on IBM Heron: hardware accuracy (blue) tracks the noiseless
simulator (grey) up to $n{=}10$ and stays ${\ge}0.95$ through $n{=}14$, even as the raw
$\langle Z\rangle$ correlation (red) decays --- the linear read-out absorbs the noise.}
\label{fig:hardware}
\end{figure}

\begin{table}[t]\centering
\caption{IBM Heron (\texttt{ibm\_aachen}) validation of trained order-3 QHA heads across
qubit counts ($100$ inputs, $2048$ shots). \emph{sim} = noiseless statevector; \emph{hw} =
hardware; $\rho$ = sim-to-hardware $\langle Z\rangle$ correlation.}
\label{tab:hw}
\begin{tabular}{lccccc}
\toprule
qubits $n$ & transpiled depth & sim acc & hw acc & $\rho$ \\
\midrule
8  & 465 & 1.00 & 1.00 & 0.91 \\
10 & 435 & 1.00 & 1.00 & 0.94 \\
12 & 590 & 1.00 & 0.96 & 0.79 \\
14 & 999 & 1.00 & 0.95 & 0.77 \\
\bottomrule
\end{tabular}
\end{table}

\subsection{Further experiments and ablations}\label{sec:more}
Fig.~\ref{fig:analysis} collects six additional studies (parity unless noted; $3$ seeds).

\paragraph{(a) Scaling with sequence length.} Fixing $k{=}3$ and sweeping
$n\in\{8,10,12,14\}$, QHA reaches ceiling for $n\ge10$ ($1.00$; $0.73$ at $n{=}8$, where the
$256$-point cube limits the disjoint training set), while classical attention stays well below
($0.51$--$0.81$) and never reaches ceiling --- the separation is a property of the interaction
order, not of a single $n$.

\begin{figure}[t]\centering
\includegraphics[width=0.6\textwidth]{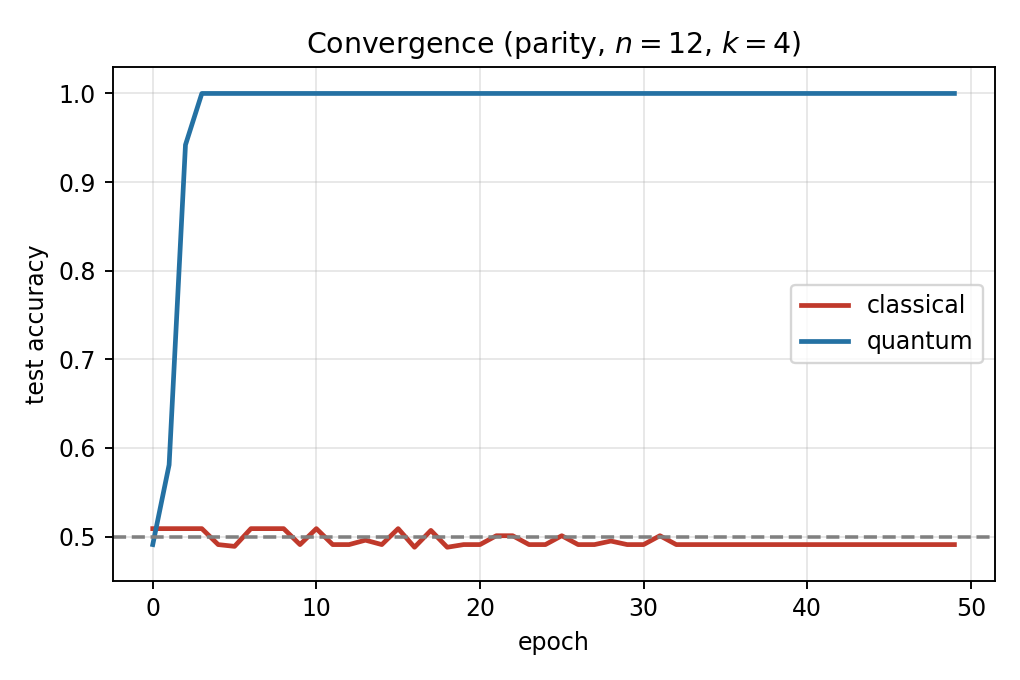}
\caption{Convergence on order-4 parity ($n{=}12$): QHA reaches ceiling within a few epochs
while the classical attention head never leaves chance at \emph{any} epoch.}
\label{fig:curves}
\end{figure}

\paragraph{(b) Depth controls reachable order.} Sweeping the number of QHA blocks
$L\in\{1,2,3,4\}$ on an order-4 task, accuracy rises sharply with depth ($0.50,\,0.83,\,1.00,
\,1.00$), matching \citet{schuld2021effect}: more re-uploads raise the reachable Fourier
degree until it covers order~$k$.

\paragraph{(c) Entanglement is load-bearing.} Ablating the $R_{ZZ}$ entangler leaves
single-qubit read-outs that depend only on individual tokens, so the head cannot represent
\emph{any} multiplicative interaction: accuracy is $\approx0.50$ for \emph{all} $k\ge2$
(Fig.~\ref{fig:analysis}c). Entanglement --- not the encoder or read-out --- supplies the
entire higher-order capacity.

\paragraph{(d) Sample efficiency.} On order-3 ($n{=}12$), QHA reaches ceiling from $\approx
500$ examples ($0.998$) while classical attention is still at chance at $1000$ and reaches only
$0.84$ at $2000$ --- a large sample-efficiency gap, consistent with the $2^{\Omega(k)}$-norm
penalty of Prop.~\ref{prop:norm}.

\paragraph{(e) Trainability vs.\ topology: narrowing the expressivity--trainability gap.} We
measure gradient variance versus qubit count and \emph{fit} its decay for three entangler
topologies, asking of each whether it \emph{also} reaches the order-$k$ target
(Fig.~\ref{fig:analysis}e). \textbf{(i) all-to-all} (our expressive default) decays
\emph{exponentially} ($\mathrm{Var}\!\sim\!2^{-0.9\,n}$) --- a barren plateau, the price of its
volume-law class --- yet it reaches the order-$k$ ceiling. \textbf{(ii) nearest-neighbour ring}
(the bounded-range variant Thm.~\ref{thm:bp} covers) decays only \emph{polynomially}
($\mathrm{Var}\!\sim\!n^{-0.30}$, no plateau), confirming the theorem, but is too sparse to route
an arbitrarily-placed order-$k$ support at logarithmic depth and stays at chance for $k\ge3$.
\textbf{(iii) log-diameter butterfly} ($O(n\log n)$ couplings at distances $1,2,4,\dots$)
\emph{attains both}: it reaches the ceiling (parity $k{=}3,4$: $1.00$ over $3$ seeds, matching
all-to-all) while its gradient variance decays far more slowly than all-to-all (best-fit
polynomial, $\approx\!30\times$ larger signal at $n{=}12$: $3.9\!\times\!10^{-3}$ vs.\
$1.3\!\times\!10^{-4}$). So a single \emph{sparse, local} variant is simultaneously expressive
enough to separate and far better conditioned than the dense head, \emph{empirically} closing
the expressivity--trainability gap; Thm.~\ref{thm:bp} rigorously anchors the trainable (ring)
end of this spectrum, while the butterfly's milder decay is measured, not yet proven.

\paragraph{(f) Noise robustness.} Under per-gate depolarizing noise (simulated), the trained
order-3 head keeps full accuracy up to $p{=}0.02$ and falls to chance only by $p{=}0.05$,
consistent with the noise tolerance observed on \texttt{ibm\_aachen}.

\paragraph{Is the advantage ``quantum'' or ``high-order''? A controlled baseline suite.}
To locate the gap we compare, at a \emph{matched} ($\sim\!300$) budget, the strongest
alternatives one would reach for (Table~\ref{tab:baselines}). The picture is clean and we
report it honestly. First, the order barrier is a property of the \emph{attention mechanism}:
\textbf{every attention-paradigm model fails for $k\ge4$} at a matched budget --- not only a
vanilla layer and a Performer, but also the explicit high-order attention variants one would
propose (a polynomial-kernel attention and a HOT-style higher-order attention). Second, the
capability is not quantum-only: explicit high-order \emph{non-attention} models (a degree-5
HOFM, a dense MLP) do reach order-$k$. What singles out QHA is therefore precise and
defensible: it is the \textbf{only attention-paradigm model that reaches order-$k$}
(Thm.~\ref{thm:sep}), it does so at the \emph{smallest} budget ($296$ params), and it is
\emph{order-adaptive} --- depth sets the reachable order, so it stays at ceiling at $k{=}5,6$
where even the param-matched tiny MLP decays ($0.51$ at $k{=}5$) and HOFM must be told the degree.

\begin{figure}[t]\centering
\includegraphics[width=\textwidth]{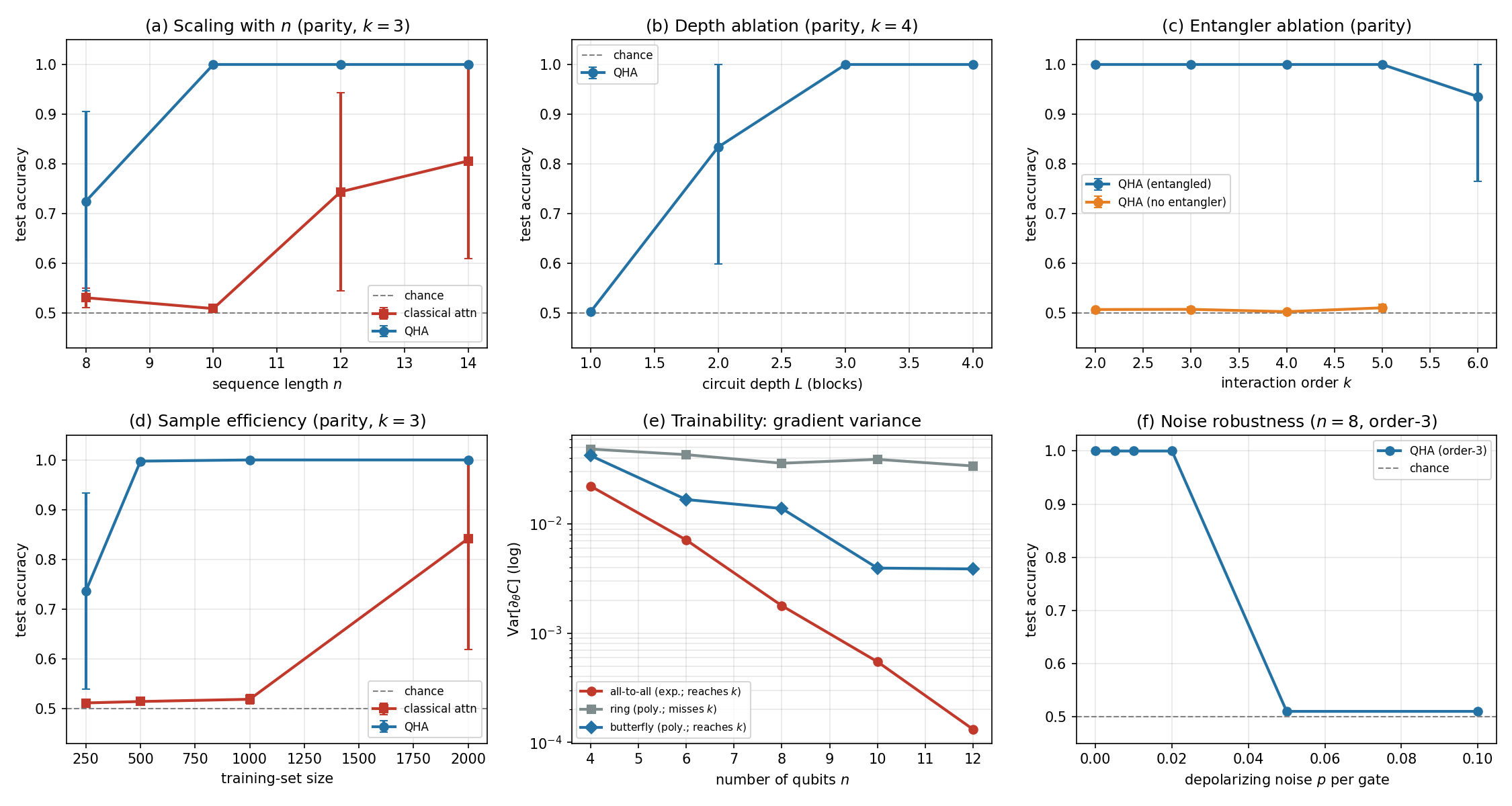}
\caption{Extended studies. (a) scaling with $n$; (b) depth vs.\ reachable order; (c) entangler
ablation; (d) sample efficiency; (e) gradient variance vs.\ qubits for three entanglers ---
all-to-all is exponential (barren) but reaches order-$k$; the ring is polynomial but misses
order-$k$; the log-diameter butterfly is polynomial-fit \emph{and} reaches order-$k$, narrowing
the expressivity--trainability gap; (f) depolarizing-noise robustness. See Sec.~\ref{sec:more}.}
\label{fig:analysis}
\end{figure}

\begin{table}[t]\centering
\caption{Controlled baseline suite (parity, $n{=}12$, test accuracy, mean$\pm$std over $6$ seeds;
$8$ for the 1-layer attention and QHA rows). \emph{Top:} \textbf{every attention-paradigm model
fails for $k\ge4$ at a matched budget} --- including explicit high-order attention
(polynomial-kernel and HOT-style) --- so the order barrier is a property of the attention
mechanism, not of tuning. \emph{Middle:} explicit high-order \emph{non-attention} models
(HOFM, MLP) do reach order-$k$, so the capability itself is not quantum-only. \emph{Bottom:}
QHA is the \textbf{only attention-paradigm model that reaches order-$k$}, and does so at the
smallest budget and without a pre-specified degree. The last column reports efficiency
$\mathrm{Eff}_{k5}=(\text{acc}-0.5)/\text{params}\times10^3$ at the hardest order $k{=}5$, scored
only for models that actually solve it (acc$\,\ge0.9$); among those, QHA is the most efficient
($2\times$ HOFM, $31\times$ the 3-layer MLP), and the param-matched tiny MLP --- tied with QHA
at $k{\le}4$ --- drops out entirely by $k{=}5$. (The degree-2 FM falls a hair below $0.5$:
chance within sampling noise, as it has no order-$\ge3$ capacity to fit.)}
\label{tab:baselines}
\begin{tabular}{lccccc}
\toprule
model & params & $k{=}3$ & $k{=}4$ & $k{=}5$ & Eff$_{k5}$ \\
\midrule
\multicolumn{6}{l}{\emph{attention-paradigm mixers (incl.\ explicit high-order attention)}}\\
1-layer softmax attention & 1922 & 0.71{\scriptsize$\pm$.23} & 0.56{\scriptsize$\pm$.09} & 0.54{\scriptsize$\pm$.08} & --- \\
Performer / linear attention & 1074 & 0.79{\scriptsize$\pm$.16} & 0.62{\scriptsize$\pm$.06} & 0.54{\scriptsize$\pm$.05} & --- \\
polynomial-kernel attention & 347 & 0.76{\scriptsize$\pm$.18} & 0.57{\scriptsize$\pm$.05} & 0.51{\scriptsize$\pm$.01} & --- \\
higher-order attention (HOT-style) & 634 & 0.83{\scriptsize$\pm$.19} & 0.53{\scriptsize$\pm$.05} & 0.51{\scriptsize$\pm$.01} & --- \\
2-layer Transformer & 4706 & 0.88{\scriptsize$\pm$.19} & 0.58{\scriptsize$\pm$.16} & 0.51{\scriptsize$\pm$.01} & --- \\
\midrule
\multicolumn{6}{l}{\emph{explicit high-order, non-attention}}\\
factorization machine (deg.\ 2) & 289 & 0.46{\scriptsize$\pm$.01} & 0.47{\scriptsize$\pm$.01} & 0.46{\scriptsize$\pm$.01} & --- \\
tiny MLP (param-matched) & 302 & 1.00{\scriptsize$\pm$.00} & 1.00{\scriptsize$\pm$.00} & 0.51{\scriptsize$\pm$.07} & --- \\
HOFM (degree-5) & 595 & 1.00{\scriptsize$\pm$.00} & 1.00{\scriptsize$\pm$.00} & 1.00{\scriptsize$\pm$.00} & 0.84 \\
3-layer MLP & 9282 & 1.00{\scriptsize$\pm$.00} & 1.00{\scriptsize$\pm$.00} & 1.00{\scriptsize$\pm$.00} & 0.05 \\
\midrule
\textbf{QHA (attention-paradigm, quantum)} & \textbf{296} & \textbf{1.00}{\scriptsize$\pm$.00} & \textbf{1.00}{\scriptsize$\pm$.00} & \textbf{1.00}{\scriptsize$\pm$.00} & \textbf{1.69} \\
\bottomrule
\end{tabular}
\end{table}

\section{Application: detecting higher-order epistasis}\label{sec:app}
The mechanism suggests a concrete use: a \emph{parameter-efficient, order-adaptive detector of
high-order feature interactions}. We develop this in the setting where high-order interaction
is the defining challenge of the field --- \emph{higher-order epistasis} in statistical
genetics --- and then confirm it transfers to two further domains.

\paragraph{Setup (pure higher-order epistasis).} Following the standard GAMETES-style
benchmark protocol, we simulate $n$ SNPs with genotypes $\{0,1,2\}$ drawn under
Hardy--Weinberg equilibrium; the phenotype is determined by a $k$-locus penetrance rule with
\textbf{no main effects} --- each causal locus is marginally uninformative, and only the joint
$k$-way genotype carries signal (\emph{pure} epistasis). We use minor-allele frequency
$\approx0.29$ so carrier status is balanced (no marginal leakage), vary the label noise
(a heritability proxy), and place the $k$ causal loci among $n{=}12$ candidate SNPs.

\paragraph{Why low-order methods fail by construction.} With no main effects, single-locus
association --- the workhorse of genome-wide studies --- is at chance, and so is any linear or
additive model: there is, by design, no order-$<k$ signal to fit. Detecting the phenotype
\emph{requires} modelling the order-$k$ genotype interaction. This is exactly the structure
QHA targets: it reads out the $k$-way interaction within a compact circuit, with the reachable
order set adaptively by depth rather than fixed a priori.

\paragraph{Phenotype prediction.} Table~\ref{tab:epi} reports balanced accuracy vs.\ order
$k$ at a fixed label noise of $0.1$ (which caps the achievable accuracy near $0.90$). The
linear/single-locus baseline and a pairwise (degree-2) model are at chance for $k\ge3$ --- the
pairwise model already succeeds at $k{=}2$ but, as expected, collapses to chance once the
interaction exceeds its degree. The field-standard \textbf{MDR} (exhaustive $k$-way contingency
search) and the classical high-order models (HOFM, MLP) succeed at all $k$. Because every
high-order detector saturates the same noise-limited ceiling, accuracy alone hides the
difference; the separating axis is \emph{efficiency} (last row of Table~\ref{tab:epi} and
Fig.~\ref{fig:application}a). \textbf{QHA reaches the ceiling with the fewest learned
parameters} ($296$; $1.5\times$ fewer than HOFM, $31\times$ fewer than the MLP), giving the
highest above-chance accuracy per parameter of any model that solves the task --- while MDR
pays an exhaustive $O(n^k)$ enumeration and the attention-paradigm baseline degrades at higher
order. 

\begin{table}[t]\centering
\caption{Higher-order epistasis --- phenotype prediction (balanced accuracy, mean$\pm$std over
6 seeds, $n{=}12$, label noise $0.1$). Pure epistasis (no main effects) $\Rightarrow$
linear/pairwise at chance for $k\ge3$; MDR is the field-standard exhaustive search. All
high-order detectors hit the same accuracy ceiling, so we read the advantage off the last
row: \textbf{detection efficiency} $\mathrm{Eff}=(\text{acc}-0.5)/\text{params}\times10^3$ at
the hardest order $k{=}4$ --- above-chance accuracy per thousand learned parameters. \emph{Among
models that reach the ceiling}, QHA is the most efficient: $1.5\times$ HOFM and $31\times$ the
MLP. Linear/pairwise (fail) and MDR (exhaustive, no learned parameters) are not scored.
\textit{(``p'' = learned parameters. Full noise sweep in App.~\ref{app:setup}.)}}
\label{tab:epi}
\begin{tabular}{lcccccc}
\toprule
order & linear & pairwise & MDR & HOFM & MLP & \textbf{QHA} \\
& {\scriptsize 26 p} & {\scriptsize 289 p} & {\scriptsize exhaustive} & {\scriptsize 451 p} & {\scriptsize 9.3k p} & {\scriptsize \textbf{296 p}} \\
\midrule
$k{=}2$ & 0.52{\scriptsize$\pm$.01} & \textbf{0.90}{\scriptsize$\pm$.00} & 0.90{\scriptsize$\pm$.00} & 0.90{\scriptsize$\pm$.00} & 0.90{\scriptsize$\pm$.00} & \textbf{0.90}{\scriptsize$\pm$.00} \\
$k{=}3$ & 0.52{\scriptsize$\pm$.01} & 0.52{\scriptsize$\pm$.01} & 0.90{\scriptsize$\pm$.00} & 0.90{\scriptsize$\pm$.00} & 0.90{\scriptsize$\pm$.01} & \textbf{0.90}{\scriptsize$\pm$.00} \\
$k{=}4$ & 0.52{\scriptsize$\pm$.01} & 0.51{\scriptsize$\pm$.01} & 0.90{\scriptsize$\pm$.01} & 0.90{\scriptsize$\pm$.00} & 0.89{\scriptsize$\pm$.01} & \textbf{0.90}{\scriptsize$\pm$.00} \\
\midrule
\rowcolor{blue!7}
\textbf{Eff.}$_{k{=}4}$ & --- & --- & --- & 0.88 & 0.04 & \textbf{1.34} \\
\bottomrule
\end{tabular}
\end{table}

\begin{figure}[t]\centering
\includegraphics[width=\textwidth]{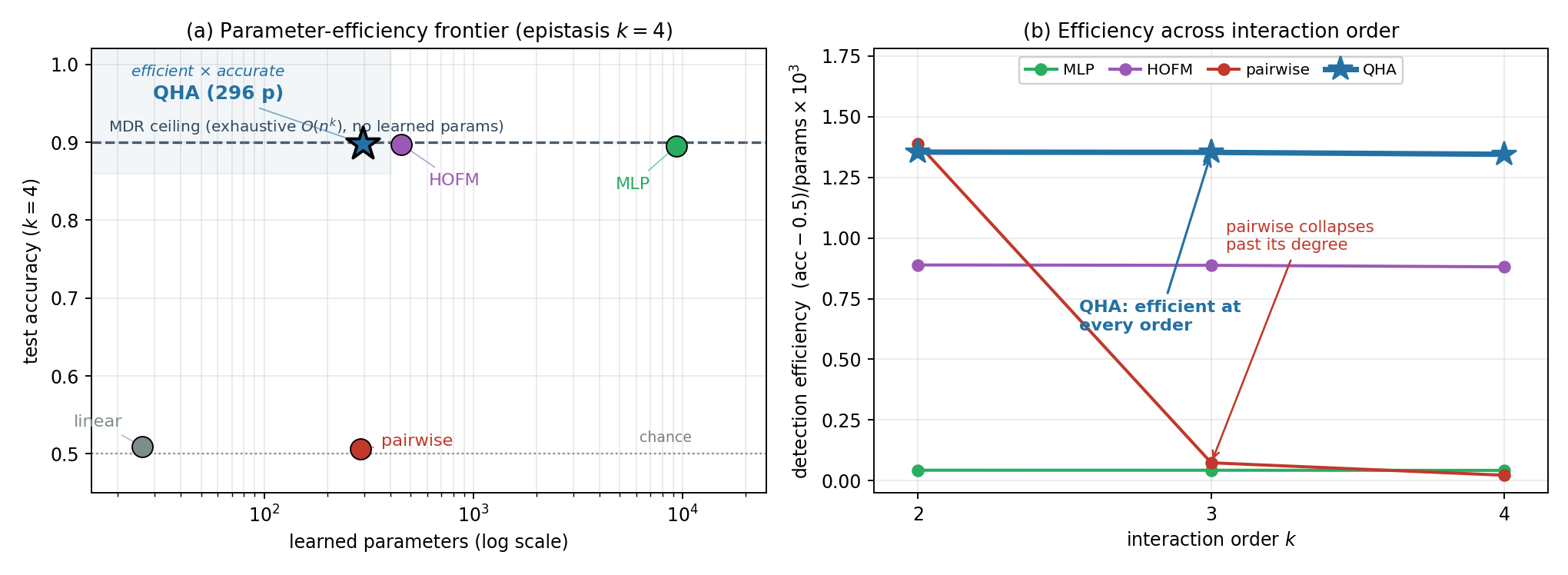}
\caption{QHA as a compact high-order detector (epistasis, noise $0.1$, 3 seeds).
(a)~\textbf{Parameter-efficiency frontier} at the hardest order $k{=}4$: all high-order
detectors hit the same accuracy ceiling, so the axis that separates them is the
parameter budget. QHA (star) sits in the ``efficient $\times$ accurate'' corner --- ceiling
accuracy with the fewest learned parameters ($296$); the MLP reaches the same accuracy at
$\approx\!31\times$ the parameters, while MDR attains it only through exhaustive $O(n^k)$
search and linear/pairwise sit at chance. (b)~\textbf{Efficiency across interaction order}: detection
efficiency $(\text{acc}-0.5)/\text{params}\times10^3$ vs.\ $k$. QHA is the only model that stays
efficient at \emph{every} order (flat at $\approx\!1.35$); the pairwise model is briefly
competitive at its own degree $k{=}2$ but collapses for $k\ge3$, HOFM is uniformly lower
($\approx\!0.88$), and the MLP sits near zero. (Linear and MDR omitted: a tiny-parameter
artifact and no learnable parameters, respectively.)}
\label{fig:application}
\end{figure}

\paragraph{Causal support recovery (interaction ranking).} Prediction alone is not
``detection'': the genetics question is \emph{which} loci interact. We rank loci by gradient
attribution --- the mean $|\partial(\text{logit gap})/\partial x_i|$ over held-out data --- and
report whether the top-$k$ recover the true causal set (exact-recovery rate and
precision@$k$), alongside MDR's natively selected set (Table~\ref{tab:recovery}). Up to
$k{=}3$, QHA recovers the interacting loci \emph{exactly} (rate $1.0$, precision@$k$ $1.0$),
matching MDR and the classical high-order learners --- identifying \emph{which} loci interact,
not merely the phenotype, which is what makes ``detector'' apt. At $k{=}4$ its \emph{prediction}
stays at the ceiling while a \emph{first-order} gradient ranking resolves the support only
partially (precision@$k\approx0.67$): exact localization of a $4$-way set is better served by an
interaction-aware (second-order) attribution than by a linear sensitivity probe, since the
order-$k$ signal --- present in the head, as the unchanged prediction shows --- carries little
first-order mass. We therefore report support recovery primarily through $k{=}3$, where the
first-order probe already pins the exact set. 

\begin{table}[t]\centering
\caption{Causal support recovery on $k$-way epistasis ($n{=}12$, 3 seeds, noise $0.1$): does
the method identify the true interacting loci? Exact = recovered set equals the true
$k$-subset; prec@$k$ = fraction of true loci in the top-$k$. Neural/FM attributions use mean
$|\partial(\text{logit gap})/\partial x_i|$. QHA recovers the support exactly through $k{=}3$;
at $k{=}4$ first-order attribution degrades while prediction does not (see text).}
\label{tab:recovery}
\begin{tabular}{lcccc}
\toprule
& MDR & HOFM (attr.) & MLP (attr.) & \textbf{QHA (attr.)} \\
\midrule
exact-recovery rate ($k{=}3$) & 1.00 & 1.00 & 1.00 & \textbf{1.00} \\
precision@$k$ ($k{=}3$) & 1.00 & 1.00 & 1.00 & \textbf{1.00} \\
\midrule
exact-recovery rate ($k{=}4$) & 1.00 & 1.00 & 1.00 & 0.00 \\
precision@$k$ ($k{=}4$) & 1.00 & 1.00 & 1.00 & 0.67 \\
\bottomrule
\end{tabular}
\end{table}

\paragraph{Intended deployment.} QHA is \emph{not} a genome-wide exhaustive scanner; it is a
high-order detector for a \emph{candidate} locus set produced by upstream screening --- the
regime where its compactness and order-adaptivity pay off.

\paragraph{Cross-domain confirmation.} The same compact detector transfers to two further
high-order tasks (Table~\ref{tab:app}): \emph{learning parity with noise} (LPN; coding/
cryptography), where the label is a noisy $\mathrm{GF}(2)$-linear function of a hidden support,
and \emph{triangle detection} in graphs with a fixed edge count (so the linear ``edge count''
is uninformative and the order-3 edge arrangement carries the label). In both, linear methods
fail, attention is weaker, and QHA reaches the ceiling at the smallest budget. The efficiency
pattern is consistent across \emph{every} experiment in the paper: on the main separation suite
(Table~\ref{tab:baselines}), epistasis (Table~\ref{tab:epi}), LPN and triangle
(Table~\ref{tab:app}), among the models that solve the task QHA always attains the highest
above-chance accuracy per parameter --- typically $1.5\times$ HOFM and $\approx30\times$ a
3-layer MLP.

\begin{table}[t]\centering
\caption{Cross-domain confirmation (test accuracy, mean$\pm$std over 6 seeds): QHA matches the
best high-order model at the lowest parameter count. The last column is detection efficiency
$\mathrm{Eff}=(\text{acc}-0.5)/\text{params}\times10^3$ for the three models that reach the
ceiling --- QHA is the most efficient in every domain.}
\label{tab:app}
\small\setlength{\tabcolsep}{4pt}
\begin{tabular}{lcccccc}
\toprule
domain (task) & linear & attn & HOFM & MLP & \textbf{QHA} & Eff \tiny(QHA/HOFM/MLP) \\
\midrule
Coding: LPN ($k{=}4$, $5\%$ noise) & 0.51{\scriptsize$\pm$.01} & 0.67{\scriptsize$\pm$.20} & 0.95{\scriptsize$\pm$.00} & 0.95{\scriptsize$\pm$.00} & \textbf{0.95}{\scriptsize$\pm$.00} & \textbf{1.52}\,/\,1.00\,/\,0.05 \\
Graphs: triangle detection & 0.68{\scriptsize$\pm$.01} & 0.75{\scriptsize$\pm$.08} & 0.97{\scriptsize$\pm$.03} & 1.00{\scriptsize$\pm$.00} & \textbf{1.00}{\scriptsize$\pm$.00} & \textbf{2.29}\,/\,0.84\,/\,0.05 \\
\bottomrule
\end{tabular}
\end{table}

\section{The Trainability--Simulability Trade-off}\label{sec:dequant}
We are explicit about what we do and do not claim. \citet{cerezo2025does} show that loss
landscapes provably free of barren plateaus often admit classical simulation after a
polynomial data-acquisition phase, because the standard route to trainability confines the
dynamics to a polynomially-sized subspace. Our trainable head (Thm.~\ref{thm:bp}) lives in
exactly this benign regime, so \textbf{we do not claim an exponential quantum speedup}. Our
contribution is an \emph{expressivity / inductive-bias} separation \emph{against classical
attention architectures}: at a matched (indeed smaller) parameter budget, QHA represents and
generalizes order-$k$ token interactions that a single classical attention layer provably
cannot (Thms.~\ref{thm:sep}, \ref{lem:matchk}; Prop.~\ref{prop:norm}). Classical
simulability of the head is, for our purposes, a \emph{feature}: it lets us validate the
separation at scale on GPUs and positions the IBM Heron run as a faithfulness demonstration
rather than a speedup claim. The all-to-all, non-Clifford ($R_{ZZ}$) entangler with
single-qubit re-uploading places the head outside the Clifford, matchgate/free-fermion, and
low-bond-dimension families, so the relevant inductive bias is genuinely quantum even though
small instances remain simulable.

\section{Conclusion}
We introduced Quantum Higher-Order Attention, a shallow quantum attention head whose local
single-qubit read-out realizes order-$k$ token interactions inside the circuit. We proved an
expressivity separation against single-layer classical attention (which is order-2),
established a barren-plateau-free trainability guarantee, and verified the prediction
empirically: at a $6.5\times$ smaller parameter budget, QHA generalizes hidden-subset parity
of every order $k\le6$ from disjoint data while classical attention degrades sharply past
$k{=}2$. A trained order-3 head transfers to IBM Heron hardware at full accuracy under a
disclosed shot budget. We are deliberate about scope: the advantage is one of
\emph{expressivity and inductive bias} against classical attention --- largest for
high-Fourier-degree targets --- not a claim of exponential speedup, and small instances of
the head remain classically simulable. QHA is a high-order \emph{module} whose value is
realised when high-order interactions matter --- which we demonstrate concretely: as a
compact high-order detector it reaches the noise ceiling on epistasis, LPN, and triangle
detection at the smallest parameter budget, where the field-standard linear methods fail
(Sec.~\ref{sec:app}). We stress the intended deployment: QHA is \emph{not} a genome-wide
exhaustive scanner but a high-order detector applied to a \emph{candidate} feature/locus set
produced by upstream screening. Two natural extensions remain: (i) scaling to larger candidate
sets via sparse candidate-set screening and hardware-aware circuit compilation, and
(ii) sharpening Lemma~\ref{lem:matchk} into an unconditional order-$k$ communication lower
bound.

\paragraph{Scope.} Hidden-subset parity and random $k$-juntas are the standard testbeds for
high-order and compositional learning in learning theory and in the study of reasoning in
sequence models \citep{bhattamishra2023simplicity,abbe2022merged,sanford2023representational}; we use them because they
isolate the one axis our theory concerns --- interaction order. The hardware study is a
faithfulness demonstration (simulation-trained parameters executed on IBM Heron for
inference), showing the learned order-$k$ rule survives device noise; it is not a
hardware-training or speedup claim (Sec.~\ref{sec:dequant}).

\subsubsection*{Reproducibility statement}
All circuits, training code, the saved hardware checkpoint, and the Qiskit hardware-evaluation
script are released. Appendix~\ref{app:setup} gives the exact architecture, data generation,
hyperparameters, and the IBM Heron protocol (backend, shots, transpilation), and every
reported number is produced by the released scripts from a fixed set of seeds.

\bibliography{iclr2026_conference}
\bibliographystyle{iclr2026_conference}

\appendix
\section{Experimental details}\label{app:setup}

\paragraph{Data.} For sequence length $n$ and order $k$, a hidden support $S\subset[n]$,
$|S|=k$, is drawn per seed. Inputs are $X\in\{-1,+1\}^n$. The label is either
\emph{parity}, $y=\prod_{i\in S}x_i$, or a \emph{$k$-junta}, $y=T(x_S)$ for a fixed balanced
random truth table $T:\{0,1\}^k\to\{0,1\}$ (chance $=0.5$). Train and test inputs are
\textbf{disjoint}: for small $n$ we enumerate the Boolean cube and partition it; otherwise we
sample distinct rows and split. This makes the metric the generalization of the order-$k$
rule, not memorization. Main experiments use $n{=}12$ ($2500$ train / $1500$ test from the
$4096$-point cube); robustness runs use $n{=}16$ ($4000$/$2000$).

\paragraph{Classical baseline.} One-layer multi-head softmax self-attention: per-token linear
embedding of the scalar bit to $d_{\text{model}}{=}16$, learned positional embedding, $2$
heads, residual + LayerNorm + a 2-layer GELU MLP, mean-pooled to a linear classifier
($1922$ parameters at $n{=}12$).

\paragraph{QHA head.} $n$ qubits; $L{=}3$ blocks (hardware head: $L{=}2$). Each block:
$H$ then $R_Z(\theta^{\text{enc}}_{\ell,w}\,x_w)$ per qubit (data re-uploading); $R_{ZZ}
(\theta^{\text{ent}}_{\ell,(i,j)})$ on all $\binom{n}{2}$ pairs; $R_Y(\theta^{\text{rot}}
_{\ell,w})$ per qubit. Read-out $\langle Z_w\rangle$ for all $w$, then a linear head to 2
logits. Parameters: $L(2n+\binom{n}{2})$ in the circuit plus the linear head
($296$ at $n{=}12,L{=}3$). Simulated with PennyLane \texttt{default.qubit} (statevector,
backprop); Adam, lr $8{\times}10^{-3}$, cosine schedule, batch $128$, $50$ epochs. The
attention-vs-QHA separation (Table~\ref{tab:main}) uses $7$--$8$ seeds; the broader baseline
and application studies use $3$. We report each run's best-epoch accuracy on the \emph{disjoint}
test set. Because the classical head stays at chance at \emph{every} epoch for $k\ge3$ (and QHA
reaches the ceiling within a few epochs without subsequent decay, Fig.~\ref{fig:curves}), the
separation reflects representational capacity rather than epoch selection: best- vs.\
last-epoch reporting leaves the gap essentially unchanged.

\paragraph{IBM Heron protocol.} The trained head's parameters are baked into the equivalent
Qiskit circuit (PennyLane and Qiskit share gate conventions: $R_Z$, $R_{ZZ}{=}$\,IsingZZ,
$R_Y$, $H$). We evaluate $\langle Z_w\rangle$ with \texttt{EstimatorV2} on
\texttt{ibm\_aachen} (156-qubit Heron r2), optimization level 2, $2048$ shots, on $100$
held-out inputs; the trained linear head is applied to the measured expectations. We report
hardware accuracy, the noiseless-simulator accuracy on identical inputs, and the
sim-to-hardware $\langle Z\rangle$ correlation.

\paragraph{Application data --- epistasis.} We simulate $n{=}12$ SNPs with genotypes
$\{0,1,2\}$ drawn i.i.d.\ under Hardy--Weinberg equilibrium at minor-allele frequency
$p{=}0.29$ (cell probabilities $(1{-}p)^2,\,2p(1{-}p),\,p^2$). A causal set $S$, $|S|{=}k$, is
drawn per seed; carrier status at locus $i$ is $c_i{=}\mathbf{1}[g_i\ge1]$ and the phenotype is
the carrier parity $y{=}\big(\sum_{i\in S}c_i\big)\bmod 2$, then label-flipped with probability
equal to the noise level (a heritability proxy). With $p{=}0.29$ each carrier indicator is
near-balanced, so there are \textbf{no main effects} and no order-$<k$ leakage; only the joint
$k$-way genotype is informative. Features are $z$-scored genotypes. We use
$4000$ train / $2000$ test, $45$ epochs, $3$ seeds, and report balanced accuracy. The Bayes
ceiling under label noise $\nu$ is $1{-}\nu$ ($1.00/0.90/0.80$ at $\nu{=}0.0/0.1/0.2$).

\paragraph{Application baselines and support recovery.} \emph{Linear}: logistic regression on
the $n$ genotypes. \emph{Pairwise}: a degree-2 factorization machine (rank $23$, $289$ params).
\emph{HOFM}: degree-5 higher-order FM via ANOVA kernel (rank $6$, $451$ params). \emph{MLP}:
one hidden layer of width $64$ ($9282$ params). \emph{QHA}: the same head as above, $L{=}3$
($296$ params). \emph{MDR} (field standard): for every $\binom{n}{k}$ locus subset we build the
$3^k$ genotype-cell contingency table, label each cell by its training majority, and score the
held-out accuracy; the best-scoring subset is the prediction and its native support estimate.
Support recovery for the learned models uses gradient attribution: per-locus importance
$s_i=\operatorname{mean}_x|\partial(\text{logit gap})/\partial x_i|$ on $1000$ held-out inputs,
ranked top-$k$. We report exact-recovery rate (top-$k$ equals the true set) and precision@$k$.

\paragraph{Application data --- LPN and triangle.} \emph{LPN} (learning parity with noise):
$X\in\{-1,+1\}^n$, $n{=}12$, label is the parity over a hidden $k{=}4$ support flipped with
probability $0.05$; $4000$/$2000$ split, disjoint rows. \emph{Triangle detection}: graphs on
$6$ vertices encoded by their $\binom{6}{2}{=}15$-bit edge-indicator vector with a
\textbf{fixed edge count} (so the linear ``number of edges'' feature is uninformative); the
label is whether the edge set contains a $3$-clique --- a product of three edge variables.
Both use the same model suite and training protocol as the epistasis runs.

\paragraph{Noise sweep.} Table~\ref{tab:noise} reports the $k{=}3$ epistasis sweep over label
noise. Every high-order detector (MDR, HOFM, MLP, QHA) tracks the Bayes ceiling $1{-}\nu$ as
noise rises, while the linear and pairwise models stay at chance throughout --- the separation
is a property of interaction order, not of the noise level.

\begin{table}[t]\centering
\caption{Epistasis noise sweep at $k{=}3$ ($n{=}12$, $3$ seeds): balanced accuracy vs.\ label
noise $\nu$. High-order detectors track the Bayes ceiling $1{-}\nu$; low-order models stay at
chance. Support recovery (exact / prec@$k$) is unchanged across $\nu$ for the high-order
detectors.}
\label{tab:noise}
\begin{tabular}{lcccccc}
\toprule
$\nu$ & linear & pairwise & MDR & HOFM & MLP & \textbf{QHA} \\
\midrule
$0.0$ & 0.51 & 0.52 & 1.00 & 1.00 & 1.00 & \textbf{1.00} \\
$0.1$ & 0.52 & 0.52 & 0.90 & 0.90 & 0.90 & \textbf{0.90} \\
$0.2$ & 0.51 & 0.51 & 0.80 & 0.79 & 0.79 & \textbf{0.80} \\
\bottomrule
\end{tabular}
\end{table}

\paragraph{Compute.} Circuit simulations run on CPU (PennyLane \texttt{default.qubit}); a
single $n{=}12$, $50$-epoch QHA run takes $\approx 35$ minutes. No quantum speedup is claimed
(Sec.~\ref{sec:dequant}); simulation cost is incurred only for training/validation at scale.

\section{Detailed proofs}\label{app:proofs}
This appendix is written to be self-contained for an ML audience: each result is preceded by
a short primer on the (non-ML) tool it uses. Throughout, a single classical self-attention
layer is the model class $T^{1,H}_{1,m,1,p}$ of \citet{sanford2023representational}: one layer, $H$ heads,
embedding dimension $m$, $p$-bit fixed-point arithmetic, on a length-$N$ input; we write
$\mathrm{budget}=mHp$ for its total ``size''.

\paragraph{Primer: Boolean functions, Fourier degree, and interaction order.}
Any function $f:\{-1,+1\}^n\!\to\!\mathbb{R}$ has a unique multilinear (Fourier) expansion
$f(x)=\sum_{S\subseteq[n]}\hat f(S)\,\chi_S(x)$, where $\chi_S(x)=\prod_{i\in S}x_i$ is the
parity (monomial) over the subset $S$ and $\hat f(S)\in\mathbb{R}$ are the Fourier
coefficients. The \emph{degree} of $f$ is the largest $|S|$ with $\hat f(S)\neq0$; we call
this the \emph{interaction order} because $\chi_S$ couples exactly the $|S|$ tokens in $S$
multiplicatively. Dot-product attention scores a pair $(i,j)$ via $q_i^\top k_j$, a degree-2
form in the token embeddings, so a single layer is naturally order-2; the question is what it
costs to reach order $k>2$. The hidden-subset parity used in our experiments, $\chi_S$ with
$|S|=k$, is the function with \emph{all} Fourier mass on a single degree-$k$ coefficient ---
the cleanest order-$k$ target.

\subsection{Lemma~\ref{lem:matchk}: order-$k$ classical lower bound}
\paragraph{Primer: communication complexity (the proof tool).}
In two-party communication complexity, Alice holds $a$ and Bob holds $b$ and they must
compute $g(a,b)$ exchanging as few bits as possible; the minimum worst-case number of bits is
the communication complexity of $g$. The canonical hard problem is \textsc{set-disjointness}
$\mathrm{DISJ}_N$: $a,b\subseteq[N]$ and $g=1$ iff $a\cap b=\varnothing$; it requires
$\Omega(N)$ bits (Yao), and its $r$-party number-in-hand version requires $\Omega(N/2^{O(r)})$.
The proof strategy is the standard ``simulation'' argument: (i) show a small attention layer
induces a cheap communication protocol; (ii) show the target task requires expensive
communication; (iii) conclude the layer cannot be small.

\begin{proof}[Proof of Lemma~\ref{lem:matchk}]
\textbf{Step 1 (a small attention layer $\Rightarrow$ a cheap protocol).} Partition the $N$
input tokens among the players. \citet{sanford2023representational} show that any $f\in T^{1,H}_{1,m,1,p}$ can
then be evaluated by a protocol exchanging only $O(mHp\cdot\log\log N)$ bits: each head's
softmax aggregation, at $p$-bit precision, is summarised by $O(m)$ numbers per token, so a
player can transmit its head contributions in $O(mp)$ bits, and there are $H$ heads. The bit
count is thus controlled by the layer's $\mathrm{budget}=mHp$, independent of which function is
computed.

\textbf{Step 2 (the order-$k$ task is expensive).} Consider the order-$k$ matching predicate
$\Match_k(X)_i=\mathbb{1}\{\exists\,j_1,\dots,j_{k-1}:\,x_i+\sum_{t}x_{j_t}\equiv 0\ (\mathrm{mod}\
M)\}$ with $M=N^{\Theta(1)}$ (for $k{=}2$ this is the pairwise ``does some other token complete
$x_i$ to $0$''; for general $k$ it asks for a matching $(k{-}1)$-tuple). Split the tokens into
$k{-}1$ blocks, one per player. A player can privately offset its block by a shared mask so
that ``a matching $(k{-}1)$-tuple, one element per block, sums to $0$'' becomes exactly ``the
$k{-}1$ sets intersect'' --- i.e.\ an instance of $(k{-}1)$-party set-disjointness
$\mathrm{DISJ}^{k-1}_N$ (the $k{=}3$ case is the two-party reduction of \citealp{sanford2023representational},
Thm~7). Hence computing $\Match_k$ solves $\mathrm{DISJ}^{k-1}_N$.

\textbf{Step 3 (combine).} $\mathrm{DISJ}^{k-1}_N$ needs $\Omega(N/2^{O(k)})$ bits. Matching
this against the $O(mHp\log\log N)$ bits the layer affords (Step 1) gives
\[
\mathrm{budget}=mHp=\Omega\!\Big(\tfrac{N}{2^{O(k)}\log\log N}\Big),
\]
which for every \emph{fixed} $k$ is $\Omega(N/\log\log N)$. So no single attention layer with
$mHp=o(N/\log\log N)$ computes $\Match_k$ for $k\ge3$; the $k=3$ case recovers
\citet{sanford2023representational}, Thm~7. \end{proof}

\noindent\emph{Remark (honest scope).} The only non-textbook ingredient is the multiparty
disjointness bound, whose constant degrades as $2^{-O(k)}$; for the constant-$k$ regime of our
experiments this is immaterial, but removing it to get a bound uniform in $k$ is left open
(Sec.~7). The lemma is a \emph{single-layer} statement; multi-layer lower bounds for the same
task are conjectural \citep{sanford2023representational,sanford2024transformers}, and we do not claim them.

\subsection{Theorem~\ref{thm:sep}: the quantum construction}
\paragraph{Primer: the Heisenberg picture (why we track observables).}
An expectation value can be computed by evolving the \emph{state} forward through the gates,
or equivalently by evolving the \emph{measured observable} backward: for a unitary $U$,
$\langle\psi|U^\dagger Z U|\psi\rangle$, so we may ask how $U^\dagger Z U$ looks. This
``Heisenberg picture'' is convenient because Clifford gates map a Pauli operator to another
(signed) Pauli. In particular $\mathrm{CNOT}_{i\to r}$ (control $i$, target $r$) maps the
target's $Z_r\mapsto Z_rZ_i$. Accumulating onto a read-out qubit $r$ therefore turns
$\langle Z_r\rangle$ into a \emph{product} of input bits --- a parity; choosing $r$
\emph{inside} the support (below) makes this product exactly $\chi_S$, with no ancilla.

\begin{proof}[Proof of the quantum direction]
We give an explicit parameter setting of one QHA head whose read-out $\langle Z_r\rangle$ on a
support qubit $r\in S$ equals the order-$k$ monomial $\chi_S(X)=\prod_{i\in S}x_i$ for any
$|S|=k$, on computational-basis inputs $x_i\in\{-1,+1\}$ (identify $+1\!\leftrightarrow\!|0\rangle$,
$-1\!\leftrightarrow\!|1\rangle$; the encoder's $R_Z$ angles are set so that qubit $w$ carries
$x_w$ in the computational basis).

\emph{What computes the parity.} Choose the read-out qubit to be a \emph{member of the
support}, $r\in S$ (e.g.\ $r=\min S$), and apply $\mathrm{CNOT}_{i\to r}$ for each
$i\in S\setminus\{r\}$; then measure $Z_r$. By the primer, after these $k-1$ CNOTs the
observable has become $Z_r\mapsto Z_r\prod_{i\in S\setminus\{r\}}Z_i=\prod_{i\in S}Z_i$, so on a
basis state $\langle Z_r\rangle=\prod_{i\in S}x_i=\chi_S(X)$, as claimed. Because $r$ is itself
in $S$, its own bit $x_r$ supplies the remaining factor: \emph{no ancilla is introduced} (the
read-out is one of the $n$ token qubits) and no spurious factor appears. Equivalently, qubit
$r$ accumulates the XOR (parity) of the support bits.

\emph{Why this lives in the QHA gate set.} Each CNOT is realised from QHA's gates
($H$, $R_Z$, $R_{ZZ}$, $R_Y$) via
\[
\mathrm{CNOT}_{i\to r}=(I\otimes H)\,\mathrm{CZ}_{i,r}\,(I\otimes H),\qquad
\mathrm{CZ}_{i,r}=e^{i\pi/4}\,R_{Z}^{(i)}(-\tfrac{\pi}{2})\,R_{Z}^{(r)}(-\tfrac{\pi}{2})\,
R_{ZZ}^{(i,r)}(\tfrac{\pi}{2}),
\]
and the Hadamard is obtained from the per-qubit $R_Y$ layer ($H=R_Y(\tfrac{\pi}{2})\,Z$, with
the global phase and Pauli-$Z$ absorbed into adjacent single-qubit rotations). Because the QHA
entangler applies $R_{ZZ}$ on \emph{every} pair, all pairs $(i,r)$, $i\in S\setminus\{r\}$, are
available.

\emph{Precise scope: architecture class vs.\ trained encoder.} We state exactly what is and is
not proved. Theorem~\ref{thm:sep} is a statement about the QHA \emph{architecture class}:
$n$ qubits (one per token, no ancilla), the gate set $\{H,R_Z,R_{ZZ},R_Y\}$, and a single-qubit
$\langle Z_r\rangle$ read-out. The witness above is a concrete circuit \emph{in this class} ---
every gate it uses ($H$, fixed-angle $R_Z$, and the CNOTs compiled from $R_{ZZ}$ and
single-qubit rotations as displayed) belongs to the set, at depth $O(\log k)$ with $O(k)$
two-qubit gates --- whose $\langle Z_r\rangle$ \emph{exactly} equals the Boolean monomial
$\chi_S$. It loads the binary inputs in the computational basis. The \emph{trained} head injects
inputs through the re-uploading rotation $R_Z(\theta^{\text{enc}}_{\ell w}x_w)$ rather than basis
loading; we do \emph{not} claim the two encoders are gate-for-gate identical. What is proved is
\emph{representability}: an order-$k$ read-out is realizable within the QHA architecture class at
logarithmic depth (Thm.~\ref{thm:sep}) while no sub-quadratic classical attention layer realizes
it (Lemma~\ref{lem:matchk}). That the re-uploading instance \emph{attains} this order is
consistent with re-uploading raising the reachable Fourier degree to $k$
\citep{schuld2021effect} and is what Sec.~\ref{sec:exp} confirms empirically --- the role we
assign to the experiments.

\emph{Resource count.} Accumulating the parity with a balanced binary tree of CNOTs onto
the support qubit $r$ (rather than a chain) uses $O(k)$ two-qubit gates at circuit depth $O(\log k)$. Each
tree level is a set of \emph{disjoint} (commuting) $R_{ZZ}$ couplings, which one QHA block's
all-to-all entangler realizes in parallel; hence the construction needs only $O(\log k)$ QHA
blocks (not $O(k)$). So a single QHA head \emph{represents} every order-$k$ monomial on a
single-qubit read-out, with $n$ qubits, depth $O(\log k)$ and $O(k)$ two-qubit gates. Combined with
Lemma~\ref{lem:matchk} (representation) and Prop.~\ref{prop:norm} (learning), no single
classical attention layer of budget $o(N/\log\log N)$ matches this, giving the separation.
\end{proof}

\subsection{Proposition~\ref{prop:norm}: fixed-budget generalization separation}
\paragraph{Primer: why norm controls generalization.}
For bounded-norm models, the generalization gap (test minus train error) is controlled by a
capacity measure (Rademacher complexity / covering number) that grows with the weight norm;
schematically, $\text{gap}\lesssim \widetilde{O}(\text{norm}/\sqrt{m})$ for $m$ samples. So if
representing a target \emph{requires} a large norm, a fixed-budget model cannot generalize it
from limited data --- it would need either more parameters/norm or far more samples.

\begin{proof}[Proof of Prop.~\ref{prop:norm}]
\textbf{Classical side.} \citet{edelman2022inductive} (Prop.~5.1, Lem.~B.2--B.3) show a bounded-norm
self-attention head representing a \emph{generic} $s$-sparse Boolean function (unknown
support, e.g.\ parity of an unknown $s$-subset) needs combined weight norm $2^{\Omega(s)}$ and
$\Omega(s)$ heads; only symmetric/low-degree special cases are cheap. Their covering-number
bound (Thm.~4.2/Cor.~4.5) gives Rademacher complexity $\widetilde{O}(\text{norm})$, hence
$\text{gap}=\widetilde{O}(\text{norm}/\sqrt m)$. A \emph{fixed}-budget head (fixed norm,
$O(1)$ heads) therefore cannot drive the gap to $o(1)$ on the order-$k$ rule once
$2^{\Omega(k)}$ exceeds its budget --- it cannot generalize from a strict subset of the cube,
which is exactly our disjoint-split protocol.

\textbf{Quantum side.} Data re-uploading raises the reachable Fourier degree with circuit
depth \citep{schuld2021effect}; concretely, the explicit construction of Thm.~\ref{thm:sep}
realizes the degree-$k$ monomial at depth $O(\log k)$ with $O(k)$ two-qubit gates and
$\mathrm{poly}(n)$ parameters, and crucially with \emph{no} $2^{\Omega(k)}$ norm blow-up. Thus
at a fixed polynomial budget QHA generalizes the order-$k$ rule, for $k$ beyond the classical
attention budget.

\textbf{Fourier-degree refinement (matches the experiments).} The classical bound bites in
proportion to the target's \emph{high-degree} Fourier mass. Parity puts all mass at degree
$k$ (worst case $\Rightarrow$ maximal gap, Fig.~\ref{fig:sep} left). A balanced random
$k$-junta spreads mass over degrees $\le k$; its low-degree part is learnable by the
low-degree-biased head \citep{bhattamishra2023simplicity}, so the classical model keeps up until $k$ is
large (Fig.~\ref{fig:sep} right). This is precisely the empirical trend. \end{proof}

\subsection{Theorem~\ref{thm:bp}: absence of barren plateaus}
\paragraph{Primer: what a barren plateau is.}
A variational circuit has a \emph{barren plateau} if the cost gradient's variance vanishes
exponentially in the number of qubits, $\mathrm{Var}_\theta[\partial_\mu C]=2^{-\Omega(n)}$;
then gradients are exponentially small almost everywhere and training is infeasible at scale.
\citet{cerezo2021cost} prove this is \emph{avoided} when two conditions hold: the circuit is
\emph{shallow} (depth $O(\log n)$) and the cost is built from \emph{local} observables (acting
on $O(1)$ qubits) rather than a global one. QHA's read-out $\langle Z_w\rangle$ is $1$-local
by construction, which is the key enabling property.

\begin{proof}[Proof of Thm.~\ref{thm:bp}]
Take the QHA variant whose entangling layers form an alternating block ansatz of local
$2$-designs on bounded-range qubits, at depth $L=O(\log n)$, with the $1$-local cost
$C(\theta)=\langle Z_w\rangle$. These are exactly the hypotheses of \citet{cerezo2021cost},
Thm~2, whose conclusion is an inverse-polynomial lower bound on the gradient variance,
\[
\mathrm{Var}_\theta[\partial_\mu C]\ \ge\ \Omega\!\big(1/\mathrm{poly}(n)\big),
\]
so there is no barren plateau: gradients shrink only polynomially, and the head is trainable
at scale. The locality of the read-out is what places the cost in the regime where this lower
bound (rather than the exponentially small global-cost bound) applies. \end{proof}

\noindent\emph{Scope (stated honestly).} Thm.~\ref{thm:bp} covers the shallow
local-$2$-design variant, to which \citet{cerezo2021cost} applies rigorously. The all-to-all
$R_{ZZ}$ entangler used in the experiments is more expressive and is \emph{not} a
geometrically local $2$-design, so the theorem does not directly certify it; we instead
\emph{observe} its trainability (order-2 parity fit within $4$ epochs; no plateau across
$k\le6$ at $n\le14$, Fig.~\ref{fig:analysis}e). As Sec.~\ref{sec:dequant} discusses, the same
locality/shallowness that guarantees trainability also tends to imply classical simulability
\citep{cerezo2025does}; we therefore present trainability as a \emph{theorem for the benign
variant} and an \emph{empirical observation for the expressive one}, and make no
exponential-speedup claim.

\end{document}